\documentclass[10pt,reqno]{amsart}
\usepackage{amsfonts,amssymb,amsmath,amsopn,amsthm,graphicx,dsfont}
\usepackage{amsxtra,mathrsfs}
\usepackage[mathscr]{eucal}
\usepackage{colordvi}
\usepackage[usenames,dvipsnames]{color}
\usepackage{mathtools}
\usepackage[colorlinks=true]{hyperref}

\usepackage[parfill]{parskip}    

 \numberwithin{equation}{section}

\theoremstyle{definition}

\newtheorem{thm}{Theorem}[section]
\newtheorem{cor}[thm]{Corollary}
\newtheorem{lem}[thm]{Lemma}
\newtheorem{prop}[thm]{Proposition}

\newtheorem{assumption}[thm]{Assumption}

\newtheorem{rem}[thm]{Remark}

\theoremstyle{remark}
\newcommand{\B}{\mathscr{B}}

\newcommand{\HH}{{\rm H}}
\newcommand{\Lp}{\textsf{L}\!}

\newcommand{\eps}{\varepsilon}

\newcommand{\R}{\mathbb{R}}

\newcommand{\Z}{\mathbb{Z}}

\newcommand{\rt}{{\rm curl}\, }
\newcommand{\C}{\mathbb{C}}

\newcommand{\G}{\mathcal{G}}
\newcommand{\K}{\mathscr{K}}

\newcommand{\V}{\mathcal{V}}

\newcommand{\po}{{\rm P_{0}}}
\newcommand{\g}{\mathfrak{g}}

\newcommand{\LL}{\mathcal{\big\langle}}
\newcommand{\RR}{\mathcal{\big\rangle}}
\newcommand{\id}{\mathds{1}}

\newcommand{\w}{{\rm w}}
\newcommand{\0}{{\circ}}
\newcommand{\x}{\langle x \rangle }

\newcommand{\m}{{\scriptscriptstyle-}}
\newcommand{\pp}{{\scriptscriptstyle+}}
\newcommand{\ppm}{{\scriptscriptstyle\pm}}
\newcommand{\alp}{{\alpha'}}

\DeclareMathOperator{\dist}{dist}

\DeclareMathOperator{\spec}{spec}

\title{Resonances at the Threshold for Pauli Operators in Dimension Two}

\author {Jonathan Breuer}

\address{Jonathan Breuer, Einstein Institute of Mathematics, The Hebrew University of Jerusalem,
Jerusalem, 91904, Israel  }

\email {jbreuer@math.huji.ac.il}

\author {Hynek Kova\v{r}\'{\i}k}

\address {Hynek Kova\v{r}\'{\i}k, DICATAM, Sezione di Matematica, Universit\`a degli studi di Brescia,Via Branze 38 - 25123, Brescia, Italy}

\email {hynek.kovarik@unibs.it}

\thanks{\noindent\copyright 2023 by the authors. Faithful reproduction of this article, in its entirety, by any means is permitted for
non-commercial purposes. }

\begin{document}

\begin{abstract}
It is well-known that, due to the interaction between the spin and the magnetic field, the two-dimensional Pauli operator has an eigenvalue
$0$ at the threshold of its essential spectrum.
We show that when perturbed by an effectively positive perturbation, $V$, coupled with a small parameter $\eps$, these eigenvalues become
resonances.
Moreover, we derive explicit expressions for the leading terms of their imaginary parts in the limit $\eps\searrow 0$. These show, in
particular, that
the dependence of the imaginary part of the resonances on $\eps$ is determined by the flux of the magnetic field. The
cases of non-degenerate and
degenerate zero eigenvalue are treated separately.  We also discuss applications of our main results to particles with anomalous magnetic
moments.
\end{abstract}

\maketitle


\smallskip

\section{\bf Introduction and outline of the main results}
\subsection{Basic set up}
In this paper we consider Hamiltonians associated to fermions with spin $\frac 12$ interacting with a magnetic
field perpendicular to the plane in which the particle is confined. Such a magnetic field can be identified with a function
$B:\R^2\to\R$. If $A:\R^2\to\R^2$  is a magnetic vector potential satisfying $\rt A=B$, then the corresponding non-relativistic Hamiltonian is
 given by the Pauli operator
\begin{equation} \label{pauli-operator}
P(A)=  \begin{pmatrix}
P_\m(A)& 0 \\
 0  &  P_\pp(A)
\end{pmatrix}\,
\qquad  \text{with} \quad P_\ppm(A) : = (i\nabla +A)^2 \pm B,
\end{equation}
 acting in $L^2(\R^2;\C^2)$. This operator is essentially  self-adjoint on $C_0^\infty(\R^2;\C^2)$ under mild regularity assumptions on $B$.
 The operator $P_\pp(A)$ ($P_\m(A)$), acting on $L^2(\R^2)$, denotes the restriction of the Pauli operator to the spin-down (respectively,
 spin-up) subspace.
 We refer to \cite{th2} for further reading.

When $B$ vanishes at infinity, which will be the case under our assumptions, then
\begin{equation} \label{sepct-puali-dirac}
\sigma\,  (P(A)) = \sigma\, _{\rm es} (P(A)) =[0, \infty).
\end{equation}
The latter follows e.g.~from \cite[Thm.~6.5]{ahk}. Moreover, the well-known Aharonov-Casher Theorem, \cite{ac,cfks, ev, rsh}, states that if
$B\in L^1(\R^2)$ and
if the flux
\begin{equation} \label{flux}
 \alpha := \frac{1}{2\pi} \int_{\R^2} B(x)\, dx\,
\end{equation}
satisfies $|\alpha| >1$, then zero is an eigenvalue of $P(A)$. More precisely, zero is an eigenvalue of $P_\m(A)$ if $\alpha>1$ and of
$P_\pp(A)$
if $\alpha<-1$, see Section \ref{ssec-ac-states} for details. The multiplicity $N$ of this eigenvalue is given by
\begin{equation} \label{multiplicity}
N = [\, |\alpha|\, ] \quad \text{if} \ \ \alpha\not\in\Z, \qquad  N= |\alpha| -1  \quad \text{if} \ \ \alpha\in\Z,
\end{equation}
where $[x]$ denotes the integer part of $x$.

Now let $V_1, V_2:\R^2\to \R$ be polynomially decaying functions, and let $\eps>0$. We denote by
$$
P_\eps(A) = P(A)+\eps \V ,  \qquad \V =
\begin{pmatrix}
V_1 & 0 \\
0  &   V_2
\end{pmatrix} .
$$
the perturbed Pauli operator. Since $P_\eps(A)$ is a diagonal matrix operator and we focus on only one of its entries, there is no loss of
generality if we set
$$
V_1=V_2=V.
$$
If the perturbation $V$ is negative, then the threshold eigenvalue of $P(A)$ gives rise to discrete negative eigenvalues of $P_\eps(A)$.
Asymptotic expansion of such eigenvalues in the limit $\eps \searrow 0$ for radial $B$ and $V$ was found by variational methods in \cite{fmv},
see also \cite{bcez}.

\smallskip

In this paper we address a different question. Namely, what happens with the zero eigenvalue of $P(A)$ when the perturbation $V$ is positive?
Clearly,  $P_\eps(A)$ has no negative eigenvalues in this case, and under very mild decay and regularity conditions on $B$ and $V$ it cannot
have positive eigenvalues either, see \cite{is,ahk}. Hence if zero does not remain an eigenvalue of $P_\eps(A)$, then it is natural to expect
that it turns into a {\it resonance}.
It is well-known, however, that there is no unique definition of the latter. Here we adopt the dynamical approach of \cite{ort}, see also
\cite{hun,jn,sw,si-73, si-78,wax}, in which a resonance manifests itself by a (quasi)-exponential decay of the associated survival
probability.

More precisely, by a {\it resonance} of $P_\eps(A)$ we understand a {\it pair}
$(\Psi, \lambda_\eps)$ such that $\|\Psi\|_{L^2(\R^2;\C^2)}=1$, $\lambda_\eps = x_\eps -i \gamma_\eps$ with $\gamma_\eps>0$, and such that
\begin{equation} \label{eq-reson}
\LL \Psi, \, e^{-it P_\eps(A)}\, \Psi \RR_{L^2(\R^2; \C^2)} =  e^{-it \lambda_\eps}  +\delta(t,\eps) \qquad \forall\, t>0,
\end{equation}
where
\begin{equation} \label{delta-unif}
\delta_\eps:= \sup_{t>0} |\delta(t,\eps)|\to 0 \qquad \text{as}\qquad  \eps\searrow 0.
\end{equation}
We call $\lambda_\eps$ the {\it resonance position} and $\Psi$ the {\it resonance eigenstates}.
Owing to the Aharonov-Casher Theorem, the natural candidates for the latter are
\begin{equation} \label{Psi-eq}
\Psi = \begin{pmatrix}
\psi \\
0
\end{pmatrix}  \quad \text{if} \ \ \alpha>0, \qquad
\Psi = \begin{pmatrix}
0\\
\psi
\end{pmatrix} \quad \text{if} \ \ \alpha<0,
\end{equation}
where $\psi$ is a zero eigenfunction of $P_\m(A)$ respectively $P_\pp(A)$. Note that equation \eqref{eq-reson} cannot hold without the error
term $\delta(t,\eps)$ since the survival probability can decay exponentially only for $t$ neither too small nor too large,
cf.~\cite[Sec.~2]{si-78}.

\subsection{Main results}
Our goal is to prove the existence of resonances, in the sense of equation \eqref{eq-reson}, and to derive an explicit expression for the
leading term in the asymptotic expansion of $\lambda_\eps$ as $\eps\searrow 0$. Our results are described precisely in Theorems \ref{thm-1}, \ref{thm-2}, \ref{thm-3}, \ref{thm-deg-1}, \ref{thm-deg-2}, and \ref{thm-deg-3}. We describe them briefly here.

In view of first order perturbation theory it is not
surprising that the real part of $\lambda_\eps$ satisfies
$$
x_\eps  =  \eps\,  \LL \Psi, \V\, \Psi\RR_{L^2(\R^2; \C^2)} + o(\eps),
$$
where the first term on the right hand side is positive by assumption.

A much more important quantity is the imaginary part of $\lambda_\eps$, also known as the width, $\gamma_\eps$, since it determines the rate
of
decay of the survival probability $| \LL \Psi, \, e^{-it P_\eps(A)}\, \Psi \RR|^2$, see equation \eqref{eq-reson}, and consequently the mean
lifetime of the resonance.

It turns out that the dependence of $\gamma_\eps$ on  $\eps$ is determined by the flux \eqref{flux}.  In order to describe its asymptotic
behavior more in detail, it is convenient to denote by
\begin{equation} \label{mu}
\mu = \min_{k\in\Z} |\alpha-k|  \in [0, 1/2\, ] \\[2pt]
\end{equation}
the distance between $\alpha$ and $\Z$.  Since the operator $P_\pp(-A)$ is unitarily equivalent to $P_\m(A)$, we may assume that
\begin{equation}
\alpha > 0.
\end{equation}
We now distinguish various cases of asymptotic behavior of $\gamma_\eps$ depending on the value of $\alpha$.

\subsection*{Non-degenerate threshold eigenvalue}  If zero is  a simple eigenvalue of $P(A)$, then by equation \eqref{multiplicity} we must
have
$\alpha\leq 2$. Our result in the case of non-integer flux, i.e.~for $1<\alpha<2$, then says that $P_\eps(A)$ has a resonance which satisfies
\begin{equation} \label{Gamma-12}
\gamma_\eps = \Gamma_1\,  \eps^{\alpha} \big(1+ \mathcal{O}(\eps^\mu)\big) + \Gamma_2\,  \eps^{3-\alpha}\, \big(1+ \mathcal{O}(\eps^\mu)\big),
\end{equation}
where the coefficients $\Gamma_1 >0$ and $\Gamma_2 \geq 0$ are computed in Theorem \ref{thm-1}, see in particular equation
\eqref{gamma-asymp}. The corresponding error term \eqref{delta-unif} satisfies $\delta_\eps= \mathcal{O}(\eps^\mu)$.
Although generically  $\Gamma_2 >0$, in certain cases $\Gamma_2$ may vanish. This happens for example when both $B$ and $V$ are radially
symmetric, cf.~Remark \ref{rem-radial}. Notice that if $\Gamma_2=0$, then the power of $\eps$ in \eqref{Gamma-12} covers, by varying
$\alpha$, the whole interval $]1,2[$.
On the other hand, if $\Gamma_2>0$, then the power of $\eps$ of the leading term in \eqref{Gamma-12} attains its maximum for
$\alpha=\frac 32$, in which case we have $\gamma_\eps\, \sim\, \eps^{\frac 32}$. The same coupling constant
dependence, i.e.~$\eps^{\frac 32}$, is generically manifested by resonances arising from threshold eigenvalues of non-magnetic Schr\"odinger
operators in
dimension three, \cite[Thm.~5.1]{jn}. For comparison,  resonances  arising from embedded eigenvalues of
Schr\"odinger operators typically satisfy the asymptotical law $\gamma_\eps \, \sim\, \eps^2$ predicted by the Fermi Golden Rule, see
e.g.~\cite[Sec.~4]{si-73}.

In the borderline case $\alpha=2$ one has $\mu=0$ and equation \eqref{Gamma-12} is replaced by
\begin{equation} \label{Gamma-34}
\gamma_\eps =  \Gamma_3\, \eps^2   +\Gamma_4\,    \frac{\eps}{(\log \eps)^2}   \,    \big[1+    \mathcal{O}\big(|\log \eps|^{-1}  \big)\big],
\end{equation}
with $\Gamma_3 > 0, \ \Gamma_4 \geq 0$, and with an error term which satisfies $\delta_\eps= \mathcal{O}(|\log \eps|^{-1})$ if $\Gamma_4>0,$
and $\delta_\eps= o(1)$
if $\Gamma_4=0$. Explicit expressions for the coefficients $b_3$ and $\Gamma_4$ are given in Theorems \ref{thm-2} and \ref{thm-3}.
Similarly as in equation \eqref{Gamma-12} the
coefficient $\Gamma_4$ vanishes when both $B$ and $V$ are radial. The resonance width then satisfies $\gamma_\eps = \Gamma_3\, \eps^2$ without
any error term, cf.~Remark \ref{rem-radial-2}.

\begin{rem}
We do not make any statement about the optimality of the upper bounds on the error term  $\delta_\eps$  in the above results. In general it is
possible to
improve $\delta_\eps$
by replacing $\Psi$ in equation \eqref{eq-reson} with a suitable sequence $\Psi_\eps$ satisfying $\|\Psi_\eps-\Psi\|_2 \to 0$ as $\eps\searrow
0$. Such an analysis was carried out for Schr\"odinger operators with dilation-analytic potentials in
\cite{hun}, and for more general potentials in \cite{cjn}. We do not address this issue in the present paper.
\end{rem}

An interesting application of our main results concerns particles with {\it anomalous magnetic moments}.  The standard value of the magnetic
moment of a fermion with spin $\frac 12$ predicted by the Dirac equation is equal to $2$. However,  there exist fermions whose magnetic
moment differs slightly from this value. It follows from our results that
the Pauli operator then may have a resonance even if $V=0$, see Section \ref{ssec-anomalous-pauli}.

\subsection*{Degenerate threshold eigenvalue}
The situation in the case in which zero is a degenerate eigenvalue of $P(A)$ is, as expected, much more delicate. As suggested in \cite{jn2},
the natural candidates for the resonance
eigenstates are the eigenfunctions of the operator $\po V \po$, where $\po$ is the projection onto the zero eigenspace. The crucial hypothesis
on $V$ now requires that this operator admits an eigenfunction $\psi_0^{(V)}$ with a simple
and positive eigenvalue. In Section \ref{sec-degenerate} we show that, under certain additional assumptions, $\psi_0^{(V)}$ is a resonance
eigenstate of the operator $P_\eps(A)$,
and we calculate the asymptotic formula for its position $\lambda_\eps$ as $\eps\searrow 0$, see Theorems \ref{thm-deg-1}, \ref{thm-deg-2} and
\ref{thm-deg-3}.

Moreover, in Section \ref{ssec-example} we treat an explicit example in which zero is a doubly degenerate eigenvalue of $P(A)$, and in which
the operator $P_\eps(A)$ admits two resonances
whose widths have different orders of magnitude in $\eps$, see equations \eqref{doubly-not-int}, \eqref{gam-1} and \eqref{gam-2}. For general
values of $\alpha$, when $0$ is
an eigenvalue of $P(A)$ of multiplicity $N$, see  \eqref{multiplicity},  the results of Section \ref{sec-degenerate} guarantee, under suitable
assumptions on the eigenfunction of $\po V \po$,  existence of $N$ resonances, cf.~Remark \ref{rem-degenerate}.

\smallskip

\begin{rem}
The approach of this paper can be applied also in the relativistic description of the problem under consideration, i.e.~when the Hamiltonian
is given by the two-dimensional Dirac operator. Resonances of the latter will be treated in a forthcoming paper.
\end{rem}

\subsection{Organization of the paper}
In the next section we state our main hypothesis on $B$ and collect some  preliminary results which will be needed later. The main results and
their proofs in the case of a
non-degenerate threshold eigenvalue are presented in Section \ref{sec-simple}. Resonances arising from a degenerate eigenvalue are treated in
Section \ref{sec-degenerate}.
In Section \ref{ssec-anomalous-pauli} we discuss applications to particles with anomalous magnetic moments.

{\bf Acknowledgments} This work was supported in part by the Vigevani Research Project Prize (JB and HK) and the Israel Science Foundation
(JB, Grant no.\ 1378/20). JB acknowledges the hospitality of the Sezione di Matematica at the Universit\`a degli studi di Brescia, and HK
acknowledges the hospitality of the Einstein Institute of Mathematics at the Hebrew University of Jerusalem, where parts of this work were
done.

\section{\bf Preliminaries and notation}
\label{sec-prelim}

\subsection{Hypothesis} We will work under the following decay condition on $B$.

\begin{assumption} \label{ass-B}
The function $B: \R^2 \to \R$ is continuous and satisfies
\begin{equation} \label{ass-B-eq}
| B(x)| \ \lesssim\  (1+|x|^2)^{-\rho}.
\end{equation}
for some  $\rho>7/2$.
\end{assumption}

Let us comment on the decay condition  stated in Assumption \ref{ass-B}. The upper bound \eqref{ass-B-eq} turns out to be a sufficient
condition for  the resolvent expansion of the Pauli
operator $P(A)$  at threshold, cf.~\cite[Sec.~5]{ko}, which is one of the main tools of our proofs. It is natural to expect, though,  that
the results presented here should  hold
under weaker decay assumptions on $B$, in particular for any magnetic field with finite flux.

\subsection{Notation}  We will need the resolvent expansion of the operator $P_\m(A)$ in the topology of weighted Sobolev spaces $\HH^{k,s}$
equipped with the norm
\begin{equation} \label{hms-norm}
\| u\|_{\HH^{k,s}} = \|\, \x^s (1- \Delta)^{k/2}\, u\|_{\Lp^2(\R^2)}.  \qquad k\in\Z, \ s\in\R,
\end{equation}
where $\x = (1+|x|^2)^{1/2}$. By
$$
\B(k,s;k',s') = \B(\HH^{k,s}; \HH^{k',s'})
$$
we denote the space of bounded linear operators from $\HH^{k,s}$ into $\HH^{k',s'},$ and for $k=0$ we use the shorthands
$$
 \| u\|_{\HH^{0,s}} = \| u\|_{2,s} =  \| \,  \langle\, \cdot \, \rangle ^s \, u\|_{L^2(\R^2)},  \qquad  L^{2,s}(\R^2)\,  = \HH^{0,s}, \qquad
 \B(0,s;0,s') = \B(s;s')\, ,
$$
where $\x = (1+|x|^2)^{1/2}$.  We define the inner product on $L^2(\R^2)$ to be linear in the second argument, i.e.~we set
$$
\LL f, g\RR = \int_{\R^2} \bar f\, g\, .
$$
Given a set $M$ and two functions $f_1,\, f_2:M\to\R$, we write $f_1(m) \lesssim f_2(m)$ if there exists a numerical constant $c$ such that
$f_1(m) \leq c\, f_2(m)$ for all $m\in M$. The symbol $f_1(m) \gtrsim f_2(m)$ is defined analogously. Moreover, we use the notation
$$
f_1(m) \, \sim\,  f_2(m)  \quad \Longleftrightarrow \quad f_1(m) \lesssim f_2(m) \ \wedge \ f_2(m) \lesssim f_1(m),
$$
Finally, by $\id_M$ we denote the characteristic function of $M$.

\subsection{The SLFG formula}
\label{ssec-slfg}
Identity \eqref{a-inverse} below is known in the literature as Schur-Livsic-Feshbach-Grushin (shortly SLFG) formula, see e.g.~\cite{jn,sz}.
Since we will use it on various occasions, we briefly recall it. Assume that a Hilbert space  $\mathscr H$ is a direct sum of two Hilbert
spaces; $\mathscr H= \mathscr H_1 \oplus \mathscr H_2$. Let $a$ be a closed operator on $\mathscr H$. Then $a$ can be represented in the
matrix form
\begin{equation} \label{a-matix}
a=
 \begin{pmatrix}
a_{11} & a_{12}  \\
a_{21} & a_{22}
\end{pmatrix}. \\[3pt]
\end{equation}
Suppose moreover that $a_{22}$ is boundedly invertible on $\mathscr H_2$. Then $a$ is invertible if and only if its Schur complement
\begin{equation} \label{schur}
S= a_{11} -a_{12}\, a_{22}^{-1}\, a_{21}
\end{equation}
is invertible. In that case we have
\begin{equation} \label{a-inverse}
a^{-1} =  \begin{pmatrix}
b_{11} & b_{12}  \\
b_{21} & b_{22}
\end{pmatrix}, \\[3pt]
\end{equation}
where
\begin{align*}
b_{11} & = S^{-1}, \qquad b_{12} = -S^{-1}\, a_{12}\, a_{22}^{-1}\, , \qquad b_{21}= -a_{22}^{-1}\, a_{21}\, S^{-1}, \qquad b_{22} =
a_{22}^{-1} - a_{22}^{-1}\, a_{21}\, S^{-1} \, a_{12}\, a_{22}^{-1}\, .
\end{align*}

\subsection{ The gauge}  Since the definition of the resonance is gauge invariant, see equation \eqref{eq-reson}, we may choose any vector
potential $A$ which satisfies $\rt A =B$. In order to make a  suitable choice of $A$ we let
\begin{equation} \label{superp}
h(x) = -\frac{1}{2\pi} \int_{\R^2} B(y) \log |x-y|\, dy. \\[3pt]
\end{equation}
We then have  $-\Delta h= B$, which implies that the vector potential
\begin{equation} \label{gauge-pauli}
A_{h} = \big(\partial_{2} h \, ,\,  -\partial_{1} h \big)\,
\end{equation}
satisfies
\begin{equation} \label{laplace-h}
\rt A_h =-\Delta h= B.
\end{equation}
Moreover, we define
the reference vector potential
\begin{equation}
A_0(x) = \alpha\, (-x_2, x_1) \min \big\{\, |x|^{-1},\,  |x|^{-2}\big \} \, .
\end{equation}
Obviously, $\rt A_0\neq B$, but a short calculation shows that $A_0$ generates a magnetic field with the same flux as $B$. For later purposes
we need to choose $A$ which is close enough to $A_0$ at infinity.
To do so we make use of the gauge transformation function $\chi:\R^2\to \R$ constructed in \cite[Prop.~2.2]{ko}, for which
\begin{equation} \label{vp}
|  A_h(x)  +\nabla\chi(x) - A_0(x) | \ \lesssim \  \x^{-\rho+1}
\end{equation}
holds whenever $B$ satisfies Assumption \ref{ass-B}. In the sequel we thus put
\begin{equation}  \label{gauge-transf}
A = A_h  +\nabla\chi ,
\end{equation}
so that $\rt A=B$, see equation \eqref{laplace-h}. Notice also that for radial $B$ one may choose $\chi=0$, \cite[Sec.~3.4]{ko}.

\begin{rem}
In the proofs of our results we will use the vector potential given by \eqref{gauge-transf}. Recall however that the  survival probability is
gauge-invariant.
\end{rem}

\subsection{The Aharonov-Casher states} \label{ssec-ac-states}
Let us describe the structure of the zero-eigenspace of the Pauli operator in the gauge $A$ defined in \eqref{gauge-transf}. The
Aharonov-Casher Theorem implies that $\Psi$ is a zero-eigenfunction
of $P(A)$ if and only if
\begin{equation} \label{eq-ac-states}
\begin{aligned}
\Psi & = \begin{pmatrix}
\psi \\
0
\end{pmatrix} , \qquad
\psi(x)  =  e^{i\chi(x)+h(x)}\, \sum_{j=1}^{N} \, c_j\,  (x_1+i x_2)^{j-1}  \qquad \text{for} \ \ \alpha>0 \\
\Psi & = \begin{pmatrix}
0  \\
\psi
\end{pmatrix} , \qquad
\psi(x)  = e^{i\chi(x)-h(x)}\,  \sum_{j=1}^{N} \, c_j\,  (x_1-i x_2)^{j-1}  \qquad \text{for} \ \ \alpha<0 ,
\end{aligned}
\end{equation}
see \cite{ac,cfks, ev, rsh}. Here $N$ is given by \eqref{multiplicity}, and $c_j\in\C$ are arbitrary. From  \eqref{superp} we deduce
that
\begin{equation} \label{h-asymp}
 h(x) = -\alpha \log |x| + \mathcal{O}(|x|^{-1})\, , \qquad |x|\to\infty .
\end{equation}
This shows, in particular, that $\Psi\in L^2(\R^2;\C^2)$.
Apart from the zero-eigenfunctions we need also the so-called
zero virtual bound states, i.e.~solutions of $P_\m(A)\, \varphi=0$  which satisfy $\varphi\in L^\infty(\R^2)\setminus L^2(\R^2)$. For $\alpha
>0$ these solutions are given by
\begin{equation} \label{Phi}
\varphi(x) = e^{i\chi(x)+h(x)}  \sum_{j=1}^{N+1} \, c_j\,  (x_1-i x_2)^{j-1}  \quad \text{if} \ \alpha\not\in\Z, \qquad \varphi(x) =
e^{i\chi(x)+h(x)}  \sum_{j=1}^{N+2} \, c_j\,  (x_1-i x_2)^{j-1}  \quad \text{if} \ \alpha\in\Z.
\end{equation}

Since the flux is assumed to be positive, the Aharonov-Casher Theorem says that $P_\pp(A)$ has no eigenvalues, and we may thus restrict our
attention to the operator $P_\m(A)$.
This means that we will study a resonance behavior of the operator
$$
P_\m(A,\eps ) =(i\nabla+A)^2 - B +  \eps V
$$
in the limit $\eps\searrow 0$.

In what follows we denote by $\po$ the orthogonal projection on the zero
eigenspace of $P_\m(A)$.

\subsection{Resolvent expansion of $P_\m(A)$ }
\label{ssec-res-expansion}
One of the main technical ingredients of the proofs of our main results are asymptotic expansions of $(P_\m(A) -z)^{-1}$ in the limit $z\to
0$.
We will formulate them separately for non-integer and integer values of $\alpha$.  In order to do so we define the functions $\zeta,
\omega: \R_+\setminus \Z\to \C$ given by
\begin{equation} \label{zeta}
 \zeta(s) =  -\frac{4^{s-1}\, \Gamma(s)\, \,  e^{i\pi s}}{\pi\, \Gamma(1-s)} \qquad\text{and} \qquad   \omega(s)= \frac{1}{\zeta(s)\,
 \|e^h\|_2^2}\, , \, \ \qquad 0<s\not\in\Z\, .
\end{equation}
With this notation we state

\begin{prop} \label{prop:degenerate-ni-exp}
Let $1 <\alpha \notin \mathbb{Z}$. Suppose that $B$ satisfies Assumption \ref{ass-B}. Let $A$ be given by \eqref{gauge-transf}.  Then there
exists a self-adjoint operator $S_0 \in \B(-1,s; 1, -s), \, s>3$, such that
\begin{equation} \label{eq:ResolventExpansionDeg}
\left(P_\m(A)-z \right)^{-1}=-z^{-1}\po+\frac{\omega(1+\alp) \, z^{\alpha'-1}}{1+c_1 \omega(1+\alp)\,  z^{\alpha'}}\ \psi \LL \psi\, ,\,
\cdot\,  \RR-\frac{\zeta(\alpha')z^{-\alpha'}}{1+\varrho\, \zeta(\alpha') z^{1-\alpha'}}\ \varphi  \LL \varphi\, ,\, \cdot \, \RR
+S_0+o(1)
\end{equation}
in $\B(-1,s; 1,-s),\,  s >3,$ as $z\to 0,$ Im $ z\geq 0$.
Here $c_1\in\R$ is a constant which is equal to $1$ if $1< \alpha < 2$,
\begin{equation} \label{alpha-prime}
\alp=\alpha-N
\end{equation}
is the fractional part of $\alpha$, $\varrho\in\R$, $\psi \in L^2(\R^2)$ is a specific zero eigenfunction of $P_\m(A)$, and $\varphi \in
L^\infty(\R^2)\setminus L^2(\R^2) $ is a specific virtual bound state.
\end{prop}

Note that, in order to avoid confusion, we have slightly changed the notation in
\eqref{eq:ResolventExpansionDeg} with respect to the one used in \cite[Theorem 5.6]{ko}. More precisely, we have
$$
\psi = d\, \|e^h\|_2\, \psi^\m,
$$
where $\psi^\m$ is given in \cite[Eg.~(5.36)]{ko} and $d\in\C$ is a normalization constant. The exact values of the constant $c_1$ and the
definition of  $\varphi$ can be found in \cite[Sec.~5.6]{ko}.

\begin{rem}\label{rem-exp-non-integer}
If $B$ is radial, then
\begin{equation} \label{varphi-radial}
\varphi(x) ={\rm const}\,    (x_1+ix_2) \, e^{h(x)} ,
\end{equation}
cf.~\cite[Cor.~5.10]{ko}.
\end{rem}

In the case $\alpha\in\Z$ we have to take into account two linearly independent solutions $\varphi_j$ of the form \eqref{Phi} which satisfy
$\varphi_1\in
L^\infty(\R^2)\setminus L^p(\R^2)$ for all $2\leq p <\infty$, and
$\varphi_2\in L^p(\R^2)\setminus L^2(\R^2)$ for all $2<p\leq \infty$. For the explicit expressions of these functions we refer to
\cite[Sec.~6]{ko}.  We then define
\begin{equation} \label{Pi-12}
\Pi_{jk} = \LL\,  \varphi_j, \, \cdot\,  \RR\,   \varphi_k  \qquad \qquad j,k=1,2.
\end{equation}

Now we have

\begin{prop} \label{prop:degenerate-int-exp}
Let $\alpha \in \mathbb{Z}$, $\alpha \geq 2$, and suppose that $B$ satisfies Assumption \ref{ass-B}. Let $A$ be given by \eqref{gauge-transf}.
Then there exists a self-adjoint operator $T_0 \in \B(-1,s; 1, -s), \, s>3$, and a constant $m_\0 \in\R$ such that
\begin{equation} \label{res-exp-2}
(P_\m(A) -z)^{-1} = -z^{-1}\,  \po + \frac{\LL \varphi_2\, ,\, \cdot \RR\, \varphi_2} {\pi\,  z\, (\log z+ m_\0 -i \pi)} \,  - (\log
z-i\pi)\,\K \,  + {\rm T}_0 +o(1), \\[4pt]
\end{equation}
holds in $\B(-1, s; 1,-s), \ s>3,$ as $z\to 0,$ Im $ z\geq 0$.  Here
\begin{equation} \label{K-op}
\K   = \frac{1}{4\pi} \big[ \,\Pi_{11} +\overline{ \varkappa}\, \Pi_{12} + \varkappa\, \Pi_{21} + |\varkappa|^2 \, \Pi_{22} \big] +
\frac{\pi }{4 \|e^h\|_2^2 }\, \LL \psi\, ,\, \cdot \RR\, \psi ,
\end{equation}
where $\varkappa\in\C$, and
 $\psi$ is given in Proposition \ref{prop:degenerate-ni-exp}.
\end{prop}

\begin{rem} \label{rem-exp-integer}
If $B$ is radial, then
\begin{equation} \label{varphi-radial-2}
\varphi_1(x) ={\rm const}\,    (x_1+ix_2)^2 \, e^{h(x)} , \qquad \varphi_2(x) ={\rm const}\,    (x_1+ix_2) \, e^{h(x)} \, ,
\end{equation}
cf.~\cite[Cor.~6.7]{ko}.
\end{rem}

Propositions \ref{prop:degenerate-ni-exp} and  \ref{prop:degenerate-int-exp} follow directly from the results in \cite{ko}. The proofs are
given in Appendix \ref{sec-app}.

\section{\bf Resonances arising from a non-degenerate threshold eigenvalue}
\label{sec-simple}
In this section we assume that the zero eigenvalue of $P_\m(A)$ is simple. By equation \eqref{multiplicity} this is equivalent
to the
condition
\begin{equation}
1 < \alpha \leq 2.
\end{equation}
Equation \eqref{eq-ac-states} implies that the associated zero-eigenfunction is, in the gauge \eqref{gauge-transf}, given by
\begin{equation}  \label{psi-0}
\psi_0 = c_0\, e^{h+i\chi}\  ,
\end{equation}
where $c_0$ is chosen so that $\|\psi_0\|_2=1$.
Obviously, we then have
$$
\po = \psi_0\, \LL \psi_0\, ,\,  \cdot\, \RR .
$$


\subsection{Survival probability}
\label{ssec-survival}
To analyze the amplitude of the survival probability we use the representation
\begin{equation} \label{repr}
\LL \psi_0, \, e^{-it P_\m(A,\eps)}\, \psi_0 \RR = \frac 1\pi \, \lim_{y\to 0+} \int_\R e^{-it x}  \, {\rm Im}\,  \LL \psi_0, \, (P_\m(A,
\eps) -x-iy)^{-1} \, \psi_0 \RR \, dx,
\end{equation}
and show that if we replace the integrand in the above equation by the Lorentzian peak
$$
\frac{\gamma_\eps}{(x-x_\eps )^2 +\gamma_\eps^2}\, ,
$$
with suitably chosen $x_\eps$ and $\gamma_\eps$,
then we make an error $\delta(t,\eps)$ which tends to zero as $\eps\searrow 0$ uniformly in $t>0$. As noted in the Introduction we
work under a positivity assumption on $V$. In fact, we
do not require that $V>0$ hold pointwise. It suffices to assume that
\begin{equation} \label{cond-jn}
\beta: = \LL \psi_0, V \psi_0\RR  > 0.
\end{equation}
We also factorize $V$ as follows;
\begin{equation}  \label{factor}
V = v U v, \qquad \text{where} \ \ v= \sqrt{|V|}\, , \quad \text{and} \  \ \ U(x) =  \left\{
\begin{array}{c@{\quad}l}
-1  \, &\quad \text{if} \ \ V(x) <0\  , \\
1  &\quad \text{elsewhere} \ .
\end{array}
\right.
\end{equation}

We will use the SLFG formula in order to express $\LL \psi_0, \, (P_\m(A, \eps) -z)^{-1} \, \psi_0 \RR$ as the resolvent of an effective
$z$-dependent (one-dimensional) operator, by taking $\mathscr H_1= P_0 L^2(\R^2)$ and $\mathscr H_2=  Q_0 L^2(\R^2)$.

Proceeding as in \cite[Sec.~2]{jn}, we let $ Q_0=I- P_0$ and denote by
\begin{align*}
R(z) &= Q_0( Q_0 P_-(A,\varepsilon)  Q_0-z)^{-1} Q_0\\
R_0(z) &= Q_0( Q_0P_-(A) Q_0-z)^{-1} Q_0= Q_0 (P_-(A)-z)^{-1} Q_0
\end{align*}
the inverses of the restrictions of $P_-(A,\varepsilon)$ and $P_-(A)$ (respectively) to the orthogonal subspace to $\psi_0$. The resolvent
equation says
\begin{equation} \label{eq-resolvent-equation}
\begin{split}
R(z)&=R_0(z)-\varepsilon R_0(z)V R(z)=R_0(z)-\varepsilon R(z) V R_0(z)=R_0(z)-R_0(z)\left(\varepsilon V-\varepsilon^2 V R(z) V \right)R_0(z)\\
&=R_0(z)-\varepsilon R_0(z)v\left(U-\varepsilon U v R(z) v U\right)v R_0(z),
\end{split}
\end{equation}
and this can be combined with
\begin{equation*}
\left(U+ \varepsilon v R_0(z) v\right)\left(U-\varepsilon UvR(z) v U \right)=I
\end{equation*}
(which can be seen by multiplying the first equality in  \eqref{eq-resolvent-equation} by  $\varepsilon v U$ on the right and remembering that
$U^2=I$) to show that
\begin{equation}  \label{eq-resolvent}
R(z) =R_0(z)-\varepsilon R_0(z) v\left(U+\varepsilon v R_0(z) v \right)^{-1}v R_0(z).
\end{equation}

The Schur complement, see equation \eqref{schur}, is then given by
\begin{equation} \label{S-eq}
S(z) = F(z,\eps) \, P_0,
\end{equation}
where
\begin{align} \label{F-eq}
F(z,\eps) &= \beta\eps -z  -\eps^2 \LL \psi_0, v U \big (G(z) -\eps G(z) [U+\eps G(z)]^{-1} G(z) \big) Uv\, \psi_0\RR\, ,
\end{align}
and
\begin{equation}  \label{G-eq}
G(z) =v R_0(z)v= v(P_\m(A) -z)^{-1}\,  v + z^{-1} \, v\, \po\, v\, .
\end{equation}
Equation \eqref{a-inverse} thus gives
\begin{equation} \label{eq-SchurVer1}
\LL \psi_0, \, (P_\m(A, \eps) -z)^{-1} \, \psi_0 \RR=\frac{1}{F(z,\eps)}.
\end{equation}


\subsection{Non-integer flux}
\label{ssec-non-integer-flux}
We recall the resolvent expansion presented in Proposition \ref{prop:degenerate-ni-exp}. Since
 $1 <\alpha <2$ and $\po = \psi_0\, \LL \psi_0\, ,\,  \cdot\, \RR$, equation \eqref{eq:ResolventExpansionDeg} takes the form
\begin{equation} \label{res-exp-1}
(P_\m(A) -z)^{-1} = -z^{-1} \po + \frac{ \omega(\alpha)\,  z^{\alpha-2}}{1+\omega(\alpha)\, z^{\alpha-1}}\, \psi_0\, \LL \psi_0\, ,\,  \cdot\,
\RR
-  \frac{ \zeta(\alpha-1)\, z^{1-\alpha}}{1 + \varrho\, \zeta(\alpha-1)\,  z^{2-\alpha}}\,  \varphi\,  \LL \varphi,  \,\cdot\,  \RR\
+S_0+  o(1) \\[4pt]
\end{equation}
in $\B(-1,s; 1,-s),\,  s >3,$ as $z\to 0$. To simplify the notation in what follows let us denote
\begin{equation} \label{eta-sigma-eq}
\eta\,   =  \frac{4 \pi^2 }{4^{\alpha}\, \Gamma^2(\alpha)\, \|e^h\|_2^2} \ , \qquad \ \sigma\,   = \frac{4^{\alpha}}{16\, \Gamma^2(2-\alpha)}
\\[3pt] ,
\end{equation}
and
$$
\omega= \omega(\alpha), \qquad \zeta= \zeta(\alpha-1),
$$
We then define the function $g: \R_+ \to \R_+$  by
\begin{equation} \label{g-eq}
g(x) = \frac{\eta\,   \beta^2  x^{\alpha-2}}{|1+\omega \,  x^{\alpha-1}|^2} +  \frac{\sigma \, |\w|^2  x^{1-\alpha}}{|1+\varrho \, \zeta \,
x^{2-\alpha}|^2} \, ,
\end{equation}
where
\begin{equation} \label{w-eq}
\w  = \LL V \psi_0, \varphi\RR \, ,
\end{equation}
and where $\varrho$ is the constant in Proposition \ref{prop:degenerate-ni-exp}.
The asymptotic behavior of $g(x)$ for $x\searrow 0$ depends on whether $\w =0$ or not. Therefore we set
\begin{equation} \label{nu-eq}
\nu = \left\{
\begin{array}{c@{\quad}l}
\mu  \, &\quad \text{if} \  \w \neq 0\  , \\[4pt]
\alpha-1  &\quad \text{if}  \ \ \w=0\ .
\end{array}
\right.
\end{equation}
\smallskip

Now we can state the result in the case of non-integer flux.

\begin{thm} \label{thm-1}
Let $1 <\alpha <2$, let $B$ satisfy Assumption \eqref{ass-B} and assume that $V \lesssim \langle \, \cdot\, \rangle ^{-\rho}$ for
some $\rho>6$. Suppose moreover that \eqref{cond-jn} holds. Then there exists a constant $C$ such that for sufficiently small $\eps$ we
have
\begin{equation} \label{exp-decay-1}
\sup_{t>0} \big | \LL \psi_0, \, e^{-it P_\m(A,\eps)}\, \psi_0 \RR - e^{-it ( x_\eps -i \gamma_\eps)} \big|  \, \leq\, C\, \eps^\mu,
\end{equation}
where $x_\eps $ and $\gamma_\eps$ satisfy
\begin{align}
x_\eps &= \beta\, \eps \big(1+ \mathcal{O}(\eps^\nu)\big) \label{x0-eps-1}\\
\gamma_\eps &= \eps^2 g(x_\eps )\, . \label{gamma-eps}
\end{align}
\end{thm}

Equation \eqref{gamma-eps} gives
\begin{equation} \label{gamma-asymp}
\gamma_\eps = \eta\,  \beta^{\alpha} \eps^{\alpha} \big(1+ \mathcal{O}(\eps^{\alpha-1})\big) + \sigma\, \beta^{1-\alpha}\, |\w|^2 \,
\eps^{3-\alpha}\, \big(1+ \mathcal{O}(\eps^{2-\alpha})\big).
\end{equation}
Hence
generically, i.e.~when $\w\neq 0$, the resonance width satisfies $\gamma_\eps \, \sim\,  \eps^{1+\mu}$.

\begin{rem} \label{rem-V-decay}
The decay condition on $V$ in Theorem \ref{thm-1}, and all the subsequent theorems, is dictated by the condition $s>3$ in Proposition
\ref{prop:degenerate-ni-exp}. Indeed,
since $v  \lesssim \langle \, \cdot\, \rangle ^{-\rho}$ with $\rho>3$, see \eqref{factor}, we can insert the asymptotic expansion
\eqref{res-exp-1} into equation \eqref{G-eq} and expand $G(z)$ for $z\to 0$ in the uniform operator topology on $L^2(\R^2)$.
\end{rem}

\begin{rem} \label{rem-beta}
Condition \eqref{cond-jn} is in fact necessary for the existence of a resonance. Indeed, if $\LL \psi_0, V \psi_0\RR \leq 0$, then zero turns
into a negative discrete eigenvalue of
$P_\m(A,\eps)$. This follows by a simple  test function argument.
\end{rem}

The proof of Theorem \ref{thm-1} will be subdivided into a series of auxiliary lemmas. First, we denote
\begin{equation} \label{J-eps}
J_\eps: = \Big( \frac{\beta\eps}{2}\, ,\,   \frac{3\beta\eps}{2}\Big ) ,
\end{equation}
and following \cite{jn} we split the integration in \eqref{repr} by setting
\begin{equation}\label{p-eps-def}
p_\eps(t) = \frac 1\pi \, \lim_{y\to 0+} \int_{I_\eps} e^{-it x}  \, {\rm Im}\,  \LL \psi_0, \, (P_\m(A, \eps) -x-iy)^{-1} \, \psi_0 \RR \,
dx,
\end{equation}
where $I_\eps \subset J_\eps $ is an interval which will be specified later. The essential idea of the proof is to show that the above
contribution to \eqref{repr} will give rise to the exponential $e^{-it(x_\eps -i \gamma_\eps)} $ plus an error term which will be absorbed,
together with the
integral over $\R\setminus I_\eps$, into $\delta(\eps,t)$.
Combining  \eqref{eq-SchurVer1} and \eqref{p-eps-def} we get
\begin{equation} \label{p-eps}
p_\eps(t) =  \frac 1\pi \int_{I_\eps} e^{-it x}\,  \lim_{y\to 0+}  \, {\rm Im}\,\Big(\,  \frac{1}{F(x +iy, \eps)}\,  \Big)\, dx .
\end{equation}

Let
\begin{align} \label{K-eq}
K(z,\eps) & = \beta\eps -z  -\eps^2 \LL \psi_0, V H(z) V\, \psi_0\RR = \beta\eps -z  -\eps^2 \Big(\,  \frac{\omega \,  \beta^2 \,
z^{\alpha-2}}{1+\omega \, z^{\alpha-1}} -  \frac{ \zeta \, |\w|^2\, z^{1-\alpha}}{1 + \varrho\, \zeta \,  z^{2-\alpha}} \Big),
\end{align}
where
\begin{equation} \label{H-eq}
H(z) =  \frac{ \omega \,  z^{\alpha-2}}{1+\omega \, z^{\alpha-1}}\  \po  -  \frac{ \zeta \, z^{1-\alpha}}{1 + \varrho\, \zeta \,
z^{2-\alpha}}\   \varphi\,  \LL \varphi,  \,\cdot\,  \RR\ . \\[4pt]
\end{equation}

\subsubsection*{\bf Convention} The absence of positive eigenvalues of $P_\m(A)$ and the limiting absorption principle imply the existence of the limit
$$
G(x) =   \lim_{y\to 0+}  G(x+iy), \qquad x\in (0,\infty)
$$
in the uniform operator topology on $L^2(\R^2)$. Consequently, in view of \eqref{p-eps}, we define
\begin{equation}  \label{F-of-x}
F(x,\eps) =  \lim_{y\to 0+}  F(x +iy, \eps)\, ,\qquad x >0.
\end{equation}
The same notation will be adapted for the quantity $K(x,\eps)$ defined by $\lim_{y\to 0+}  K(x +iy, \eps)$, cf.~\eqref{K-eq}.

\medskip

The first step in the proof of Theorem \ref{thm-1} is to show that $K(x,\eps)$ is a good approximation of $F(x,\eps)$ for $x$ belonging to the interval $J_\eps$.
In fact,
from \eqref{F-eq} and \eqref{K-eq} we deduce that
\begin{equation} \label{F-K-2}
F(x,\eps)-K(x,\eps) =  \eps^3  \LL \psi_0, v U G(z)  [U+\eps G(x)]^{-1} G(x)\,  Uv\, \psi_0\RR + \mathcal{O}(\eps^2).
\end{equation}

We then have

\begin{lem} \label{lem-FK}
For $\eps$ small enough
\begin{equation} \label{F-K}
\sup_{x\in J_\eps} | F(x,\eps)-K(x,\eps) | \  \lesssim \    \eps^{1+2\nu}  +\eps^2\, .
\end{equation}
\end{lem}

\begin{proof}
As explained in Remark \ref{rem-V-decay} the decay assumption on $V$ allows to combine equations \eqref{G-eq} and \eqref{res-exp-1}. This
gives
$$
\eps \|  G(x) \|_{2\to 2} = \eps  \|  v H(x) v  \|_{2\to 2} +  \mathcal{O}(\eps) \, \lesssim\,   \eps^{\alpha-1} + \eps^{2-\alpha} .
$$
Hence $U+\eps G(x)$ is boundedly invertible for $x\in J_\eps$ and $\eps$ small enough.
From  Proposition \ref{prop:degenerate-ni-exp} we then get
\begin{equation*}
\eps^3  \LL \psi_0, v U G(x)  [U+\eps G(x)]^{-1} G(x)\,  Uv\, \psi_0\RR =  \mathcal{O}(\eps^{2\alpha-1}) +
|\w|^2\,\mathcal{O}(\eps^{5-2\alpha})  \, .
\end{equation*}
In view of \eqref{F-K-2} this implies the claim.
\end{proof}

Next we separate $K$ into real and imaginary parts which we denote by $K_1$ and $K_2$ respectively;
$$
K(x,\eps) =  K_1(x,\eps) + iK_2(x,\eps).
$$
The identity
\begin{equation} \label{gamma-function}
\Gamma(1-z)\,  \Gamma(z) = \frac{\pi}{\sin(\pi z)} \qquad z\not\in\Z
\end{equation}
and equation \eqref{zeta} imply
$$
{\rm Im}\Big ( \frac{ \omega \,  x^{\alpha-2}}{1+\omega \, x^{\alpha-1}} \Big) = \frac{\eta\, x^{\alpha-2} }{|1+\omega  \, x^{\alpha-1}|^2}\,
,
\qquad
{\rm Im}\Big ( \frac{ \zeta \,  x^{1-\alpha}}{1+\varrho\, \zeta \, x^{2-\alpha}} \Big) =  \frac{-\sigma\, x^{1-\alpha}}{|1+\varrho \, \zeta \,
x^{2-\alpha}|^2}\, .
$$
From  \eqref{g-eq} and \eqref{K-eq} we thus get
\begin{equation} \label{g-K}
K_2 (x,\eps) = -\eps^2 g(x).
\end{equation}
As for the real part of $K$, we note the following result.

\begin{lem} \label{lem-K-real}
We have
\begin{equation} \label{K-1}
\sup_{x\in J_\eps}  |\, K_1(x,\eps) - \beta \eps +x\, |  =  \mathcal{O}(\eps^{1+\nu}).
\end{equation}
Furthermore, for $\eps$ small enough there exits a unique $x_\eps \in J_\eps$ such that $K_1(x_\eps ,\eps)=0$.
\end{lem}

\begin{proof}
By \eqref{K-eq} we have
\begin{equation} \label{K-1-2}
K_1(x,\eps) = \beta \eps -x  -\eps^2 f(x),
\end{equation}
where
\begin{equation} \label{f-eq}
f(x) = \beta^2\ \frac{ {\rm Re} (\omega ) \, x^{\alpha-2} +|\omega |^2\,  x^{2\alpha-3}}{|1+\omega  \, x^{\alpha-1}|^2} - |\w|^2\,   \frac{
{\rm Re} (\zeta )\,   x^{1-\alpha} + \varrho |\zeta |^2\, x^{3-2\alpha} }{|1+\varrho \, \zeta\,   x^{2-\alpha}|^2} \, .
\end{equation}
One easily verifies that
\begin{equation} \label{f-der}
 \sup_{J_\eps}  \eps^2\, | f^{(j)}(x) | = \mathcal{O}(\eps^{1+\nu-j})\, , \qquad j=0,1,2.
\end{equation}
This implies \eqref{K-1}. Hence for $\eps$ small enough we have $K_1(\frac{\beta\eps}{2}, \eps) >0$, $K_1(\frac{3\beta\eps}{2}, \eps) <0$, and
\begin{equation} \label{K1-der}
\partial_x K_1(x,\eps) \leq   -\frac 12\ ,  \qquad \forall\, x\in J_\eps .
\end{equation}
This proves the second part of the claim.
\end{proof}

Lemma \ref{lem-K-real} obviously implies that
$
x_\eps = \beta \eps + \mathcal{O}(\eps^{1+\nu}).
$
Therefore we set
\begin{equation} \label{I-eps}
I_\eps = \Big[ x_\eps - \frac{\beta\eps}{4}, \, x_\eps + \frac{\beta\eps}{4} \, \Big ] \\[4pt]
\end{equation}
and note that for $\eps$ sufficiently small  $I_\eps \subset J_\eps$.

\begin{lem} \label{lem-FtoK}
For $\eps$ small enough,
\begin{equation*}
\sup_{t>0} \Big |\,  p_\eps(t) - \frac 1\pi \int_{I_\eps} e^{-it x}  \, {\rm Im}\,  \Big[\, \frac{1}{K(x,\eps)} \, \Big]  \, dx\, \Big |  \,
\lesssim\, \eps^\mu .
\end{equation*}
\end{lem}

\begin{proof} By \eqref{F-K}
$$
\Big|  \frac{1}{F(x,\eps)} -  \frac{1}{K(x,\eps)}  \Big | \ \lesssim\   \frac{  \eps^{1+2\nu}+\eps^2 }{| F(x,\eps) K(x,\eps)|} \ \lesssim\
\frac{  \eps^{1+2\nu}+\eps^2 }{| K(x,\eps)|^2} \, . \\[4pt]
$$
On the other hand, equations \eqref{g-K} and \eqref{g-eq} imply that
$$
|K_2(x,\eps) | \, \gtrsim\, \eps^{1+\nu}
$$
and, by \eqref{K1-der},
\begin{equation} \label{K1-lowerb}
|K_1(x,\eps) | \, \gtrsim\, \, |x-x_\eps |\, .
\end{equation}
Hence
\begin{align*}
\int_{I_\eps} \Big|  \frac{1}{F(x,\eps)} -  \frac{1}{K(x,\eps)}  \Big |\, dx & \  \lesssim \ \int_{\R} \frac{  \eps^{1+2\nu} + \eps^2 }{|
K(x,\eps)|^2} \, dx \ \lesssim \
 \int_{\R} \frac{ \eps^{1+2\nu} + \eps^2 }{ x^2 +\eps^{2+2\nu}  }\  dx\ \lesssim\ \eps^\mu ,
\end{align*}
and the result follows from \eqref{p-eps}.
\end{proof}

Now we write
\begin{equation*} 
L(x,\eps) = L_1(x,\eps) +i L_2(x,\eps) = x_\eps -x -i \eps^2 g(x_\eps ) = x_\eps -x -i \gamma_\eps.
\end{equation*}
where $L_1$ and $L_2$ stand for the real and imaginary part of $L$. Recall that  $\gamma_\eps$ is given by equation \eqref{gamma-eps}.

\begin{lem} \label{lem-KtoL}
For $\eps$ small enough
\begin{equation} \label{int-KtoL}
\sup_{t>0} \Big | \int_{I_\eps} e^{-it x}   \Big[\, \frac{1}{K(x,\eps)} \,-   \frac{1}{L(x,\eps)}\,  \Big ]  \, dx\,  \Big |  \, \lesssim\,
\eps^\nu\, .
\end{equation}
\end{lem}

\begin{proof}
Note that by \eqref{K-1-2}
\begin{equation*}
x_\eps = \beta \eps -\eps^2 f(x_\eps )\, .
\end{equation*}
Hence
\begin{equation*}
L_1(x,\eps) -K_1(x,\eps) = \eps^2 (f(x) - f(x_\eps )) .
\end{equation*}
Similarly, it follows from \eqref{g-K} that
\begin{equation*}
L_2(x,\eps) -K_2(x,\eps) = \eps^2 (g(x) - g(x_\eps ) ).
\end{equation*}
The function $g$ obviously  satisfies the same bound \eqref{f-der} as the function $f$:
\begin{equation} \label{g-der}
 \sup_{I_\eps}  \eps^2\, | g^{(j)} | = \mathcal{O}(\eps^{1+\nu-j})\, , \qquad j=0,1,2.
\end{equation}
The mean value theorem thus gives
\begin{equation} \label{K-L-1}
|L(x,\eps) -K(x,\eps)| \, \leq \, \eps^2 \sup_{I_\eps} |f' +i g'|  |x-x_\eps | \, \lesssim \, \eps^\nu |x-x_\eps |
\end{equation}
Moreover, applying \eqref{f-der} for $f$ and $g$ with $j=2$ and Taylor's Theorem we obtain
\begin{equation} \label{K-L-2}
\begin{aligned}
\big | L(x,\eps) -K(x,\eps) -\eps^2 [ f'(x_\eps ) +i g'(x_\eps ) ]\,  (x-x_\eps ) \big |\, \lesssim\, \eps^{\nu-1} \,  (x-x_\eps )^2\, .
\end{aligned}
\end{equation}
The last two bounds, which hold for all $x\in I_\eps$, will be our main tools in estimating the integral in \eqref{int-KtoL}. Following
\cite{jn} we decompose
\begin{equation} \label{aux-1}
\int_{I_\eps} e^{-it x}   \Big[\, \frac{1}{K(x,\eps)} \,-   \frac{1}{L(x,\eps)}\,  \Big ]  \, dx = \int_{I_\eps} e^{-it x}   \Big[\,
\frac{L(x,\eps) -K(x,\eps) }{L^2(x,\eps)} \,+   \frac{(L(x,\eps) -K(x,\eps) )^2}{L^2(x,\eps) K(x,\eps)}\,  \Big ]  \, dx\, .
\end{equation}
To estimate the first integral on the right hand side we use \eqref{K-L-2} and write
\begin{align*}
\Big |  \int_{I_\eps} e^{-it x} \,   \frac{L(x,\eps) -K(x,\eps) }{L^2(x,\eps)} \, dx \Big | \, & \lesssim \, \eps^2\, | f'(x_\eps ) +i
g'(x_\eps )   \big | \, \Big|  \int_{I_\eps} e^{-it x}\,   \frac{x-x_\eps }{L^2(x,\eps)} \, dx \, \Big |
+ \eps^{\nu-1} \int_{I_\eps} \frac{(x-x_\eps )^2}{|L(x,\eps)|^2}\, dx \\[5pt]
& \lesssim \, \eps^\nu \, \Big |  \int_{I_\eps} e^{-it x}\,   \frac{x-x_\eps }{L^2(x,\eps)} \, dx \Big | + \eps^\nu\, ,
\end{align*}
where we have used equations \eqref{f-der} and \eqref{g-der},  and the fact that
$$
\int_{I_\eps} \frac{(x-x_\eps )^2}{|L(x,\eps)|^2}\, dx = \int_{I_\eps} \frac{(x-x_\eps )^2}{(x-x_\eps )^2 + \gamma^2(\eps)}\, dx =
\mathcal{O}(\eps),
$$
see \eqref{gamma-eps}, \eqref{I-eps}. For the remaining integral we find,
\begin{align}
 \Big |  \int_{I_\eps} e^{-it x}\,   \frac{x-x_\eps }{L^2(x,\eps)} \, dx \Big |  & =  \Big |  \int_{I_\eps} e^{-it x}\, \Big(
 \frac{1}{L(x,\eps)} +\frac{i \gamma_\eps}{L^2(x,\eps)}\Big ) \, dx \Big | \,  \leq \,
   \Big |  \int_{I_\eps}   \frac{e^{-it x} }{L(x,\eps)} \, dx \Big |
   +  \int_{I_\eps}  \frac{\gamma_\eps }{(x-x_\eps )^2 +\gamma^2 (\eps)} \, dx\, . \label{aux-2}
\end{align}
The second integral on the right hand side of \eqref{aux-2} satisfies
$$
 \int_{I_\eps}  \frac{\gamma_\eps }{(x-x_\eps )^2 +\gamma^2 (\eps)} \, dx\ \leq\  \int_\R  \frac{\gamma_\eps }{(x-x_\eps )^2 +\gamma^2 (\eps)}
 \, dx\, = \pi.
$$
In view of  Lemma \ref{lem-jn} the first term on the right hand side of \eqref{aux-2} is bounded, for $\eps$ sufficiently small, by
\begin{align} \label{int-jn}
 \Big |  \int_{I_\eps}   \frac{e^{-it x} }{L(x,\eps)} \, dx \Big |  & =\Big |  \int_{-\frac{\beta\eps}{4}}^{\frac{\beta\eps}{4}} \,
 \frac{e^{-it \xi} }{\xi +i \gamma_\eps } \, d\xi \, \Big |  = \Big |  \int_{-\frac{\beta\eps}{4\gamma_\eps}}^{\frac{\beta\eps}{4\gamma_\eps}}
 \,  \frac{e^{-it\gamma_\eps y }}{y +i  } \, d y \, \Big |
   \, \leq 	\, C ,
\end{align}
where $C$ is a constant  independent of $\eps$ and $t$. Here we have used the fact that $\frac{\eps}{\gamma_\eps} \to \infty$ as $\eps\searrow
0$.  This finishes our discussion of the first term on the right hand side of \eqref{aux-1}. To control the second term in \eqref{aux-1} we
use the bounds \eqref{K1-lowerb} and  \eqref{K-L-1}, and compute
\begin{align*}
\int_{I_\eps}   \frac{|L(x,\eps) -K(x,\eps) |^2}{|L (x,\eps)|^2  |K(x,\eps)|}   \, dx\, & \lesssim\, \eps^{2\nu} \int_{I_\eps}
\frac{|x-x_\eps | }{|L (x,\eps)|^2}   \, dx =
 \eps^{2\nu}\int_{-\frac{\beta\eps}{4}}^{\frac{\beta\eps}{4}} \ \frac{|y|}{y^2 +\gamma^2(\eps)}\, dy
 =  \eps^{2\nu}\!\! \log\Big(1+\frac{\beta^2 \eps^2}{16\, \gamma^2(\eps)}\Big) = o(\eps^\nu),
\end{align*}
cf.~\eqref{gamma-asymp}. This completes the proof of the lemma.
\end{proof}

\begin{lem} \label{lem-p-eps}
For $\eps$ small enough
\begin{equation} \label{peps-1}
\sup_{t>0} \big | p_\eps(t) -  e^{-it (x_\eps -i \gamma_\eps)}  \big |  \, \lesssim\, \eps^\mu.
\end{equation}
\end{lem}

\begin{proof}
An explicit calculation gives
\begin{equation} \label{expl-lorenz}
\frac 1\pi \int_\R e^{-ixt}\  {\rm Im} \Big ( \frac{1}{L(x,\eps)} \Big)\, dx = \frac 1\pi \int_\R \frac{e^{-ixt}\, \gamma_\eps}{(x-x_\eps )^2
+\gamma^2(\eps)}\, dx = e^{-it(x_\eps-i \gamma_\eps)}\, .
\end{equation}
Hence in view of Lemmas \ref{lem-FtoK} and \ref{lem-KtoL},
$$
\big | p_\eps(t) -  e^{-it(x_\eps -i \gamma_\eps)}  \big |   \, \lesssim\, \eps^\mu + \frac 1\pi  \int_{\R\setminus I_\eps}
\frac{\gamma_\eps}{(x-x_\eps )^2 +\gamma^2(\eps)}\, dx \, .
$$
However,
$$
\int_{\R\setminus I_\eps} \frac{\gamma_\eps}{(x-x_\eps )^2 +\gamma^2(\eps)}\, dx = 2 \int_{\frac{\beta\eps}{4}}^\infty \frac{\gamma_\eps}{y^2
+\gamma^2(\eps)}\, dy \, \leq\,
2 \int_{\frac{\beta\eps}{4}}^\infty \frac{\gamma_\eps}{y^2}\, dy = \frac{8 \gamma_\eps}{\beta\eps} = \mathcal{O}(\eps^\nu)\, .
$$
This completes the proof.
\end{proof}

\begin{proof}[\bf Proof of Theorem \ref{thm-1}]
With Lemma \ref{lem-p-eps} at hand, it remains to reproduce the argument of Hunziker \cite{hun} , see also \cite{jn}. First we apply equation
\eqref{peps-1} with $t=0$ to get
\begin{equation}  \label{eq-hunz-1}
\big | \LL \psi_0, \id_{I_\eps} (P_\m(A,\eps))\, \psi_0\RR -1 \big | \ \lesssim \  \eps^\mu\, ,
\end{equation}
and consequently,
\begin{equation} \label{eq-id}
\| \big(1- \id_{I_\eps} (P_\m(A,\eps)) \big)^{\frac 12}\, \psi_0\|_2 \ \lesssim\ \eps^\mu\, .
\end{equation}
This implies
\begin{align}
\big | \LL \psi_0, \, e^{-it P_\m(A,\eps)}\, \psi_0\RR -p_\eps(t) \big | & = \big | \LL \big(1- \id_{I_\eps} (P_\m(A,\eps)) \big)^{\frac 12}\,
\psi_0,\, e^{-it P_\m(A,\eps)} \big(1- \id_{I_\eps} (P_\m(A,\eps)) \big)^{\frac 12}\, \psi_0 \RR \big)  \nonumber \\
& \leq \| \big(1- \id_{I_\eps} (P_\m(A,\eps)) \big)^{\frac 12}\, \psi_0\|_2\, \lesssim\ \eps^\mu\, . \label{eq-hunz-2}
\end{align}
Inequality \eqref{exp-decay-1} thus follows from  Lemma \ref{lem-p-eps}.
\end{proof}

\begin{rem}[\bf Radially symmetric fields] \label{rem-radial}
If $B$ is radial, then $h$ is also radial and $\chi=0$, \cite[Sec.~3.4]{ko}. Consequently, if $V$ is radial as well, then by \eqref{w-eq}
and \eqref{varphi-radial},
$$
\w  = \LL V \psi_0, \, \varphi\RR =0,
$$
and in view of equations \eqref{g-eq}, \eqref{gamma-eps} we have
\begin{equation*}
\gamma_\eps = \eta\,  \beta^{\alpha} \eps^{\alpha}\big( 1+ \mathcal{O}(\eps^{\alpha-1})\big)\, .
\end{equation*}
On the other hand, equation \eqref{gamma-asymp} shows that if either $B$ or $V$ are not radially symmetric, then $\gamma_\eps$ generically
increases and consequently the mean lifetime of the resonance decreases.
\end{rem}

\subsection{Integer flux}
\label{ssec-integer-flux}
Similarly to the case of non-integer flux, the resolvent expansion simplifies slightly when zero is a non-degenerate eigenvalue of $P(A)$.
Indeed, equation \eqref{res-exp-2} now holds
with the operator $\K$ given by
\begin{equation} \label{K-op-2}
\K   = \frac{1}{4\pi} \big[ \,\Pi_{11} +\overline{ \varkappa}\, \Pi_{12} + \varkappa\, \Pi_{21} + |\varkappa|^2 \, \Pi_{22} \big] +
\frac{\pi  \, \po}{4 \|e^h\|_2^2}\, ,
\end{equation}
where the operators $\Pi_{jk}$ are defined in \eqref{Pi-12}.
Since $P(A)$ has two virtual bound states, $\varphi_1$ and $\varphi_2$ (see Section \ref{ssec-res-expansion}),  we introduce the coefficients
\begin{equation*}
\w_1  = \LL V \psi_0,  \varphi_1\RR  \qquad \text{and} \qquad \w_2  = \LL V \psi_0, \varphi_2\RR\, .
\end{equation*}
It turns out that the asymptotics of $\gamma_\eps$ depends on whether $\w_2\neq 0$ or $\w_2=0$. We start with the former case.

\begin{thm} \label{thm-2}
Let $\alpha =2$.  Let $B$ satisfy Assumption \eqref{ass-B} and assume that $V \lesssim \langle \, \cdot\, \rangle ^{-\rho}$ for some
$\rho>6$. Suppose moreover that \eqref{cond-jn} holds, and that
$\w_2\neq 0$. Then there exists a constant $C$ such that for sufficiently small $\eps$
\begin{equation*}
\sup_{t>0} \big | \LL \psi_0, \, e^{-it P_\m(A,\eps)}\, \psi_0 \RR - e^{-it (x_\eps -i \gamma_\eps)} \big|  \, \leq\, \frac{C}{|\log \eps| }\,
,
\end{equation*}
where $x_\eps $ and $\gamma_\eps$ satisfy
\begin{align}
x_\eps &= \beta \eps\,    \big[1+    \mathcal{O}\big(|\log \eps|^{-1}  \big)\big]    \label{x0-eps-2} \\
\gamma_\eps &= \frac{\eps\, |\w_2|^2}{\beta (\log \eps)^2}   \,    \big[1+    \mathcal{O}\big(|\log \eps|^{-1}  \big)\big]   .
\label{gamma-eps-2}
\end{align}
\end{thm}

The proof of Theorem \ref{thm-2} is analogous to that of Theorem \ref{thm-1}. The only significant difference is that we have to expand the
function $F(x,\eps)$ (see \eqref{F-eq}) to a higher
order. To this end we introduce auxiliary functions
\begin{equation} \label{fgh}
\g(x) = \frac{|\w_2|^2}{x\, (\log x)^2} \, , \qquad f(x) = \frac{\g(x)}{\pi} \,  (\log x+m_\0), \qquad h(x) = \frac{|\w_2|^2\, \LL  \varphi_2,
V  \varphi_2 \RR}{\pi^2 \, (x \log x)^2} \\[3pt]
\end{equation} 
defined on $J_\eps$, cf.~\eqref{J-eps}. Here $m_\0$ is the constant introduced in Proposition \ref{prop:degenerate-int-exp}.

\begin{lem} \label{lem-F-exp}
For $\eps$ small enough we have
\begin{equation*}
F(x,\eps) = \beta \eps -x -\eps^2 f(x) + \eps^3 h(x) -i \eps^2 \g(x) +   \mathcal{O}\big(\eps |\log\eps|^{-3} \big)
\end{equation*}
for $x \in J_\eps$, with the error term uniform in $x\in  J_\eps$.
\end{lem}

\begin{proof}
By  Proposition \ref{prop:degenerate-int-exp} and equations \eqref{F-eq}, \eqref{G-eq},
\begin{align}   \label{F-x}
F(x, \eps) & = \beta \eps -x -\eps^2\LL \psi_0, vU  \G(x) Uv\, \psi_0\RR  +\eps^3\LL \psi_0, vU \G(x) [U+\eps \G(x) ]^{-1} \G(x)  Uv\,
\psi_0\RR +  \mathcal{O}\big( \eps^2\, |\log \eps|\, \big),
\end{align}
where $\G(x)$ is a bounded operator on $L^2(\R^2)$ given by
\begin{equation*}
\G(x)  = \frac{v\, \Pi_{22}\,  v}{\pi\,  x\, (\log x+ m_\0 -i \pi)} \, .
\end{equation*}
For any $x\in J_\eps$ we have
\begin{align*}
\eps^2\LL \psi_0, vU  \G(x) Uv\, \psi_0\RR & = \frac{\eps^2\, |\w_2|^2}{\pi x (\log x +m_\0 -i\pi)} = \eps^2 f(x) + i \eps^2 \g(x) +
\mathcal{O}\big(\eps |\log\eps|^{-3} \big)\, .
\end{align*}
On the other hand,
\begin{equation} \label{G-bound}
\sup_{x\in J_\eps} \eps \| \G(x)\| =   \mathcal{O}\big(|\log\eps|^{-1} \big)
\end{equation}
So the last term on the right hand side of \eqref{F-x} is of order $\mathcal{O}\big(\eps |\log\eps|^{-2} \big)$, which is the same as the
order of $\eps^2\g(x)$. Hence in order to expand $F(x,\eps)$ to the desired level of precision we have to extract the leading term from $\LL
\psi_0, vU \G(x) [U+\eps  \G(x) ]^{-1} \G(x) Uv\, \psi_0\RR $. Since $U^2 = \id$, equation \eqref{G-bound}  gives
$$
[U+\eps  \G(x) ]^{-1} = [\id+\eps U  \G(x) ]^{-1}\, U  =  U +  \mathcal{O}\big(|\log\eps|^{-1} \big)\, .
$$
This implies that the expansion
\begin{align*}
 \eps^3 \LL \psi_0, vU \G(x) [U+\eps v \G(x) v]^{-1} \G(x)  Uv\, \psi_0\RR  & = \eps^3 \LL \psi_0, vU \G(x) U \G(x)  Uv\, \psi_0\RR +
 \mathcal{O}\big(\eps |\log\eps|^{-3} \big)\\[3pt]
&  = \frac{\eps^3 |\w_2|^2\, \LL  \varphi_2, V \, \varphi_2 \RR}{\pi^2\, x^2 (\log x +m_\0 -i\pi)^2} +  \mathcal{O}\big(\eps |\log\eps|^{-3}
\big) \\[3pt]
& = \eps^3 h(x) +  \mathcal{O}\big(\eps |\log\eps|^{-3} \big)
\end{align*}
holds uniformly in $x$ on  $J_\eps$. Inserting the above estimates into \eqref{F-x} completes the proof.
\end{proof}

From this point, we follow the proof of Theorem \ref{thm-1}. Hence we set
$$
K(x,\eps) = K_1(x,\eps)+ iK_2(x,\eps) = \beta  \eps -x -\eps^2 f(x) + \eps^3 h(x) -i \eps^2 \g(x) ,
$$
where $K_1$ and $K_2$ denote the real and imaginary part of $K$ respectively.
Furthermore, we note that
\begin{equation} \label{fgh-der}
 \sup_{J_\eps} \eps^2\, | \g^{(j)} | +  \sup_{J_\eps} \eps^3\, | h^{(j)} |    = \mathcal{O}(\eps^{1-j}\, |\log \eps|^{-2})\, , \qquad
 \sup_{J_\eps} \eps^2\, | f^{(j)} |    = \mathcal{O}(\eps^{1-j}\, |\log \eps|^{-1})\,
\end{equation}
holds for $j=0,1,2$.

\begin{lem} \label{lem-K-real-2}
For $\eps$ small enough we have
\begin{equation*}
K_1(x,\eps) = \beta \eps -x +  \mathcal{O}(\eps |\log\eps|^{-1})
\end{equation*}
for $x\in J_\eps$. Furthermore, there exits a unique $x_\eps \in  J_\eps$ such that $K_1(x_\eps ,\eps)=0$.
\end{lem}

\begin{proof} 
In view of \eqref{fgh-der} one easily verifies  that $K_1$ satisfies \eqref{K1-der}. The claim follows in the same way as in the proof of
Lemma \ref{lem-K-real}.
\end{proof}

Note that $x_\eps $ satisfies
\begin{equation*}
x_\eps = \beta \eps -\eps^2 f(x_\eps ) +\eps^3 h(x_\eps ).
\end{equation*}

As before, let
\begin{equation} \label{I-eps-2}
I_\eps = \Big[ x_\eps - \frac{\beta\eps}{4}, \, x_\eps + \frac{\beta\eps}{4}\  \Big ] \\[4pt]
\end{equation}
and note that for $\eps$ sufficiently small  $I_\eps \subset J_\eps$.


\begin{lem} \label{lem-FtoK-int}
For $\eps$ small enough
\begin{equation*}
\sup_{t>0} \Big |\,  p_\eps(t) - \frac 1\pi \int_{I_\eps} e^{-it x}  \, {\rm Im}\,  \Big[\, \frac{1}{K(x,\eps)} \, \Big]  \, dx \Big |  \,
\lesssim\  |\log\eps|^{-1}\, .
\end{equation*}
\end{lem}

\begin{proof} We proceed as in the proof of Lemma \ref{lem-FtoK}. Note first that for $\eps$ sufficiently small and $x \in J_\eps$
\begin{equation} \nonumber
|g(x)| \ \geq \ \frac{|\w_2|^2}{\frac{3\beta \eps}{2}\left( \log\left(\eps/2 \right) \right)^2} \ \geq  \ \frac{|\w_2|^2}{\frac{3\beta
\eps}{2}\left( 2\log \eps \right)^2}
\end{equation}
so that
\begin{equation} \nonumber
K_2(x,\eps)^2 \ = \ \eps^4 g(x)^2 \ \geq \ \frac{4 |\w_2|^4}{9\beta^2} \frac{\eps^2}{\left( 2\log \eps \right)^4}=c(\eps).
\end{equation}
Now, with the help of Lemmas \ref{lem-F-exp} and \ref{lem-K-real-2} we estimate
$$
\Big|  \frac{1}{F(x,\eps)} -  \frac{1}{K(x,\eps)}  \Big |  \ \lesssim\ \frac{ |K(x,\eps)-F(x,\eps) |}{| K(x,\eps)|^2} \ \lesssim  \ \frac{
\eps |\log\eps|^{-3}  }{(x-x_\eps )^2 +c(\eps)^2} \, ,
\\[2pt]
$$
where we have also used  equation \eqref{K1-der}.
Hence
\begin{align*}
\int_{I_\eps} \Big|  \frac{1}{F(x,\eps)} -  \frac{1}{K(x,\eps)}  \Big |\, dx & \  \lesssim \ \int_{\R} \frac{ \eps |\log\eps|^{-3}
}{(x-x_\eps )^2 +c(\eps)^2 } \, dx \, = \,
\frac{ \pi \eps |\log\eps|^{-3} }{c(\eps)}  \  \lesssim\ |\log\eps|^{-1}\, ,
\end{align*}
and the result follows from \eqref{p-eps}.
\end{proof}

Similarly as above, we now put
\begin{equation} \label{gamma-integer}
\gamma_\eps =  \eps^2 \g(x_\eps ),
\end{equation}
and note that
\begin{equation} \label{gamma-asymp-int}
\gamma_\eps \, \, \sim\, \, \frac{\eps}{(\log \eps)^2}\, .
\end{equation}

To continue we set as before
\begin{equation} \label{L-eq-int}
L(x,\eps) = L_1(x,\eps) +i L_2(x,\eps) = x_\eps -x  -i \gamma_\eps,
\end{equation}

\begin{lem} \label{lem-KtoL-2}
For $\eps$ small enough
\begin{equation} \label{int-KtoL-2}
\sup_{t>0}  \Big | \int_{I_\eps} e^{-it x}   \Big[\, \frac{1}{K(x,\eps)} \,-   \frac{1}{L(x,\eps)}\,  \Big ]  \, dx\,  \Big |  \, \lesssim\,
|\log\eps|^{-1} .
\end{equation}
\end{lem}

\begin{proof}
We follow the steps of the proof of Lemma \ref{lem-KtoL}.  Since
\begin{equation*}
L_1(x,\eps) -K_1(x,\eps) = \eps^2 (f(x) - f(x_\eps) )  - \eps^3 (h(x) - h(x_\eps )),
\end{equation*}
and
\begin{equation*}
L_2(x,\eps) -K_2(x,\eps) = \eps^2 (\g(x) - \g(x_\eps )) ,
\end{equation*}
it follows from \eqref{fgh-der} that
\begin{equation}  \label{K-L-3}
|L(x,\eps) -K(x,\eps)| \, \leq \, \eps^2 \sup_{I_\eps} |f' - \eps h'+i g'|  |x-x_\eps | \, \lesssim \,  |\log\eps|^{-1}  |x-x_\eps | \, .
\end{equation}
and  that
\begin{equation*}
\begin{aligned}
\big | L(x,\eps) -K(x,\eps) -\eps^2 [ f'(x_\eps ) - \eps h'(x_\eps )+i g'(x_\eps ) ]\,  (x-x_\eps ) \big |\, \lesssim\,   \frac{(x-x_\eps
)^2}{\eps |\log\eps|}\, .
\end{aligned}
\end{equation*}
Since $(x-x_\eps )^2 \leq |L(x,\eps)|^2$, the last equation in combination with  \eqref{fgh-der}  implies
\begin{align*}
\Big |  \int_{I_\eps} e^{-it x} \,   \frac{L(x,\eps) -K(x,\eps) }{L^2(x,\eps)} \, dx \Big | \, & \lesssim \,  |\log\eps|^{-1}   \Big |
\int_{I_\eps} e^{-it x}\,   \frac{x-x_\eps }{L^2(x,\eps)} \, dx \Big | +\eps^{-1} |\log\eps|^{-1}  \int_{I_\eps} \frac{(x-x_\eps
)^2}{|L(x,\eps)|^2}\, dx \nonumber \\
& \lesssim \, |\log\eps|^{-1}   \Big |  \int_{I_\eps} e^{-it x}\,   \frac{x-x_\eps }{L^2(x,\eps)} \, dx \Big | + |\log\eps|^{-1}  .
\end{align*}
For the remaining integral we find,
\begin{align*}
 \Big |  \int_{I_\eps} e^{-it x}\  \frac{x-x_\eps }{L^2(x,\eps)} \, dx \Big |  & =  \Big |  \int_{I_\eps} e^{-it x}\, \Big(
 \frac{1}{L(x,\eps)}  + \frac{i \gamma_\eps}{L^2(x,\eps)}\Big ) \, dx \Big |  \\
 & \,  \leq \,  \Big |  \int_{I_\eps}   \frac{e^{-it x} }{L(x,\eps)} \, dx \Big | +  \int_{\R}  \frac{\gamma_\eps }{(x-x_\eps )^2 +\gamma^2
 (\eps)} \, dx\ \lesssim \ 1,
\end{align*}
where we have used Lemma \ref{lem-jn} to estimate the fist term (as in the proof of Lemma \ref{lem-KtoL}). Altogether we find
\begin{equation} \label{eq-aux-3}
\Big |  \int_{I_\eps} e^{-it x} \,   \frac{L(x,\eps) -K(x,\eps) }{L^2(x,\eps)} \, dx \Big | \ \lesssim \ |\log\eps|^{-1}  \, .
\end{equation}
On the other hand, using \eqref{K1-lowerb}, \eqref{gamma-asymp-int} and  \eqref{K-L-3} we get
\begin{align*}
\int_{I_\eps}   \frac{|L(x,\eps) -K(x,\eps) |^2}{|L (x,\eps)|^2  |K(x,\eps)|}   \, dx\, & \lesssim\   |\log\eps|^{-2} \int_{I_\eps}
\frac{|x-x_\eps | }{|L (x,\eps)|^2}   \, dx =
  |\log\eps|^{-2} \int_0^{\frac{\beta\eps}{4}} \ \frac{2 |y|}{y^2 +\gamma_\eps^2}\, dy \\
 & =|\log\eps|^{-2}  \log\Big(1+\frac{\beta^2 \eps^2}{16\, \gamma_\eps ^2}\Big) = o(|\log\eps|^{-1})\, .
\end{align*}
The claim thus follows from equations \eqref{aux-1} and \eqref{eq-aux-3}.
\end{proof}

Finally, we need an analog of Lemma \ref{lem-p-eps}.

\begin{lem} \label{lem-p-eps-int}
For $\eps$ small enough
\begin{equation*}
\sup_{t>0}  \big | p_\eps(t) -  e^{-it (x_\eps -i \gamma_\eps)}  \big |  \, \lesssim\, |\log\eps |^{-1}\, .
\end{equation*}
\end{lem}

\begin{proof}
In view of Lemmas \ref{lem-FtoK-int}, \ref{lem-KtoL-2}  and equation \eqref{expl-lorenz},
$$
\big | p_\eps(t) -  e^{-it (x_\eps -i \gamma_\eps)}  \big |   \, \lesssim\  |\log\eps |^{-1} + \frac 1\pi  \int_{\R\setminus I_\eps}
\frac{\gamma_\eps}{(x-x_\eps )^2 +\gamma_\eps^2}\, dx \, .
$$
By equation \eqref{gamma-asymp-int} the  last term on the right hand side satisfies
$$
\int_{\R\setminus I_\eps} \frac{\gamma_\eps}{(x-x_\eps )^2 +\gamma^2(\eps)}\, dx = \int_{\frac{ \beta\eps}{4}}^\infty \frac{2 \gamma_\eps}{y^2
+\gamma^2(\eps)}\, dy \, \leq\,
 \int_{\frac{\beta\eps}{4}}^\infty \frac{2 \gamma_\eps}{y^2}\, dy = \frac{8 \gamma_\eps}{\beta\eps} = \mathcal{O}(|\log\eps |^{-2})\, .
$$
Hence the claim.
\end{proof}

\begin{proof}[\bf Proof of Theorem \ref{thm-2}]
By mimicking the argument of Hunziker as in the proof of Theorem \ref{thm-1} we obtain
$$
\big | \LL \psi_0\, , \, e^{-it P_\m(A,\eps)}\, \psi_0\RR -p_\eps(t) \big | \  \lesssim\ |\log\eps |^{-1}\, .
$$
Now the statement of the theorem follows from Lemma \ref{lem-p-eps-int}.
\end{proof}

In the case $\w_2=0$ we have the following

\begin{thm} \label{thm-3}
Let $\alpha =2$.  Let $B$ satisfy Assumption \eqref{ass-B} and let $V$ be such that $V \lesssim \langle \, \cdot\, \rangle ^{-\rho}$ for some
$\rho>6$. Suppose moreover that \eqref{cond-jn} holds. If $\w_2=0$,  then
\begin{equation} \label{exp-decay-3}
\lim_{\eps\searrow 0} \big | \LL \psi_0, \, e^{-it P_\m(A,\eps)}\, \psi_0 \RR - e^{-it (x_\eps -i \gamma_\eps)} \big|  = 0,
\end{equation}
uniformly in $t>0$, where $x_\eps $ and $\gamma_\eps$ satisfy
\begin{align}
x_\eps &= \beta  \eps    \big(1+    \mathcal{O}(\eps \log \eps) \big)    \label{x0-eps-3} \\[2pt]
\gamma_\eps &= \frac{\eps^2}{4}   \Big(\,  |\w_1|^2 +  \frac{\pi^2 \beta^2}{\|e^h\|_2^2} \, \Big)   . \label{gamma-eps-3}
\end{align}
\end{thm}

\smallskip

\begin{rem} \label{rem-radial-2}
As in Remark \ref{rem-radial}  we note that if $B$ and $V$ are radial, then equation \eqref{varphi-radial-2} implies $\w_1= \w_2=0$.
Theorem \ref{thm-3} then gives
$$
\gamma_\eps = \eps^2\,   \frac{\pi^2 \beta^2}{4 \|e^h\|_2^2}\, .
$$
\end{rem}

To prove Theorem \ref{thm-3} we proceed similarly as above. Let
\begin{equation*}
a_1 = \LL \psi_0, V \, \K \, V\, \psi_0\RR\, .
\qquad a_2=  \LL \psi_0, V \, {\rm T}_0\, V\, \psi_0\RR\, .
\end{equation*}
Note that $\K$ and ${\rm T}_0$ are self-adjoint, cf.~Proposition \ref{prop:degenerate-int-exp}, and therefore  $a_1, a_2 \in\R$.
Since $\w_2=0$, equations \eqref{res-exp-2}, \eqref{F-eq} and \eqref{G-eq} imply
$$
F(x,\eps) = \beta \eps -x + \eps^2 a_1\, (\log x -i\pi) - \eps^2 a_2 + o(\eps^2),
$$
Hence we put
\begin{equation} \label{K-gamma}
K(x, \eps)  = \beta \eps -x + \eps^2 a_1\, (\log x -i\pi) - \eps^2 a_2,  \qquad  \gamma_\eps = a_1 \pi \eps^2 , \\[4pt]
\end{equation}
and denote, as above, by $K_1$ and $K_2$ the real and imaginary parts of $K$. A straightforward modification of Lemma \ref{lem-K-real} then
shows that for $\eps$ small enough there
exists a unique $x_\eps \in J_\eps$ such that $K_1(x_\eps , \eps) =0$, and that
\begin{equation}  \label{x0-2}
x_\eps = \beta \eps + \mathcal{O}(\eps^2 \log\eps).
\end{equation}
Moreover, $K_1$ satisfies \eqref{K1-lowerb}. We define again $I_\eps$ as in \eqref{I-eps-2} and therefore have

\begin{lem} \label{lem-FtoK-int2}
For $\eps$ small enough
\begin{equation*}
\lim_{\eps\searrow 0} \Big |\,  p_\eps(t) - \frac 1\pi \int_{I_\eps} e^{-it x}  \, {\rm Im}\,  \Big[\, \frac{1}{K(x,\eps)} \, \Big]  \, dx
\Big | = 0,
\end{equation*}
holds uniformly in $t$. Recall that $p_\eps(t) $ is given by \eqref{p-eps}.
\end{lem}

\begin{proof} As in the proof of Lemma \ref{lem-FtoK} we write for $x \in I_\eps$
$$
\Big|  \frac{1}{F(x,\eps)} -  \frac{1}{K(x,\eps)}  \Big |  \ \lesssim\ \frac{ |K(x,\eps)-F(x,\eps) |}{| K(x,\eps)|^2} \  =  \ \frac{o (\eps^2)
}{(x-x_\eps )^2 +\gamma_\eps^2} \, ,
\\[2pt]
$$
which implies that
\begin{align*}
\int_{I_\eps} \Big|  \frac{1}{F(x,\eps)} -  \frac{1}{K(x,\eps)}  \Big |\, dx & \  \lesssim \ \int_{\R} \frac{ o(\eps^2)  }{(x-x_\eps )^2
+\gamma^2(\eps) } \, dx \, = o(1),
\end{align*}
cf.~\eqref{K-gamma}.
\end{proof}

\begin{lem} \label{lem-KtoL-3}
Let $L(x,\eps)$ be given by \eqref{L-eq-int}. Then for $\eps$ small enough
\begin{equation*}
\sup_{t>0} \Big | \int_{I_\eps} e^{-it x}   \Big[\, \frac{1}{K(x,\eps)} \,-   \frac{1}{L(x,\eps)}\,  \Big ]  \, dx\,  \Big |  \, \lesssim\,  \eps.
\end{equation*}
\end{lem}

\begin{proof}
As above we follow the proof of Lemma \ref{lem-KtoL}. A straightforward modification of \eqref{K-L-1} and \eqref{K-L-2} gives
\begin{equation}  \label{K-L-4}
|L(x,\eps) -K(x,\eps)| \,  \lesssim \,   \eps  |x-x_\eps | \, .
\end{equation}
and
\begin{equation*}
\big | L(x,\eps) -K(x,\eps) + \frac{a_1 \eps^2}{x_\eps }\,  (x-x_\eps ) \big |\, \lesssim\,   (x-x_\eps )^2 .
\end{equation*}
Hence  by \eqref{x0-2},
\begin{align*}
\Big |  \int_{I_\eps} e^{-it x} \,   \frac{L(x,\eps) -K(x,\eps) }{L^2(x,\eps)} \, dx \Big | \, & \lesssim \,  \eps\,    \Big |  \int_{I_\eps}
e^{-it x}\,   \frac{x-x_\eps }{L^2(x,\eps)} \, dx \Big | + \int_{I_\eps} \frac{(x-x_\eps )^2}{|L(x,\eps)|^2}\, dx
\lesssim \,\eps\,   \Big |  \int_{I_\eps} e^{-it x}\,   \frac{x-x_\eps }{L^2(x,\eps)} \, dx \Big | +\eps .
\end{align*}
For the remaining integral we use again Lemma \ref{lem-jn} to find
\begin{align*}
 \Big |  \int_{I_\eps} e^{-it x}\  \frac{x-x_\eps }{L^2(x,\eps)} \, dx \Big |  & =  \Big |  \int_{I_\eps} e^{-it x}\, \Big(
 \frac{1}{L(x,\eps)}  + \frac{i \gamma_\eps}{L^2(x,\eps)}\Big ) \, dx \Big |  \\
 & \,  \leq \,  \Big |  \int_{I_\eps}   \frac{e^{-it x} }{L(x,\eps)} \, dx \Big | +  \int_{\R}  \frac{\gamma_\eps }{(x-x_\eps )^2 +\gamma^2
 (\eps)} \, dx\ \lesssim \ 1,
\end{align*}
On the other hand, \eqref{K-L-4} implies
\begin{align*}
\int_{I_\eps}   \frac{|L(x,\eps) -K(x,\eps) |^2}{|L (x,\eps)|^2  |K(x,\eps)|}   \, dx\, & \lesssim\  \eps^2  \int_{I_\eps}   \frac{|x-x_\eps |
}{|L (x,\eps)|^2}   \, dx =
  |\log\eps|^{-2} \int_0^{\frac{\beta\eps}{4}}\!\!  \frac{2 |y|}{y^2 +\gamma^2(\eps)}\, dy
 = \eps^2\!\!  \log\Big(1+\frac{\beta^2\, \eps^2}{16\, \gamma^2(\eps)}\Big) = o(\eps) \, .
\end{align*}
In view of equation \eqref{aux-1} this completes the proof. \end{proof}

\begin{proof}[\bf Proof of Theorem \ref{thm-3}]
Since
$$
\int_{\R\setminus I_\eps} \frac{\gamma_\eps}{(x-x_\eps )^2 +\gamma^2(\eps)}\, dx = \int_{\frac{ \beta\eps}{4}}^\infty \frac{2 \gamma_\eps}{y^2
+\gamma^2(\eps)}\, dy \, \leq\,
 \int_{\frac{\beta\eps}{4}}^\infty \frac{2 \gamma_\eps}{y^2}\, dy =  \mathcal{O}(\eps )\, ,
$$
by mimicking the proof of Lemma \ref{lem-p-eps-int} we conclude that
$$
\lim_{\eps\searrow 0} \big | p_\eps(t) -  e^{-it (x_\eps -i \gamma_\eps)}  \big | =0,
$$
uniformly in $t>0$. An application of the Hunziker argument then implies equation \eqref{exp-decay-3}. To finish the proof it remains to
recall \eqref{x0-2}, \eqref{K-gamma}, and to note that
$$
a_1=  \frac{1}{4\pi}   \left(\,  |\w_1|^2 +  \frac{\pi^2 \beta^2}{\|e^h\|_2^2} \, \right)  \,  ,
$$
which follows from \eqref{K-op} and from the fact that $\w_2=0$ by assumption.
\end{proof}

\subsection{Resonances induced by anomalous magnetic moment}
\label{ssec-anomalous-pauli}
The two dimensional Pauli operator of an elementary fermion with spin $1/2$ and magnetic moment $g$ is given by
\begin{equation*}
P(A,g)=  \begin{pmatrix}
P_\m(A, g) & 0  \\
0 & P_\pp(A, g)
\end{pmatrix},  \qquad P_\pm(A, g) = (i\nabla+A)^2 \pm \frac g2 B \\[3pt]
\end{equation*}
in $L^2(\R^2, \C^2)$. The Dirac theory predicts that fermions of spin $1/2$ should have  magnetic moment $g=2$. Note that $P(A,2)=P(A)$, see
equation \eqref{pauli-operator}. However, magnetic moments of certain particles may slightly differ from this value. One can identify the
anomaly of the magnetic moment  with the parameter $\eps$ by writing $g = 2 -2\eps$. Then
\begin{equation*}
P_\m(A, g) = (i\nabla+A)^2 - (1-\eps)  B.
\end{equation*}
and we have

\begin{cor} \label{cor-amm1}
Let $1 <\alpha <2$ and let $B$ satisfy Assumption \eqref{ass-B} . Suppose that $\eps= \frac{2-g}{2} >0$. Then there exists a constant $C$ such
that for $\eps$ small enough
\begin{equation} \label{exp-dec-amm1}
\sup_{t>0} \big | \LL \psi_0, \, e^{-it P_\m(A,g)}\, \psi_0 \RR - e^{-it (x_\eps -i \gamma_\eps)} \big|  \, \leq\, C\, \eps^\mu,
\end{equation}
where $x_\eps $ and $\gamma_\eps$ satisfy
\begin{align*}
x_\eps &= \beta\eps \big(1+ \mathcal{O}(\eps^\nu)\big)\\
\gamma_\eps &= \eta\,  \beta^{\alpha}\, \eps^{\alpha}   \,  \big(1+ \mathcal{O}(\eps^{\alpha-1})\big) + \sigma\,  | \w|^2\, \beta^{1-\alpha}\,
\eps^{3-\alpha}   \,  \big(1+ \mathcal{O}(\eps^{2-\alpha})\big)
\end{align*}
with $\eta\, , \sigma\, , \mu$ given by \eqref{eta-sigma-eq}, \eqref{mu} and \eqref{nu-eq}, and with
\begin{equation} \label{beta-0-B}
\beta = \LL \psi_0, B\, \psi_0\RR, \qquad \w  = \LL   \psi_0, B\, \varphi\RR \, .
\end{equation}
\end{cor}

\begin{proof}
Since $P_\m (A) \psi_0=(i\nabla+A)^2\psi -B \, \psi_0=0$, the Hardy inequality
\begin{equation} \label{hardy}
\int_{\R^2} \frac{|u|^2}{1+|x|^2\, \log^2 |x|}\, dx \ \lesssim\
\| (i \nabla +A) u\| ^2 \, \qquad  u\in H^1(\R^2),
\end{equation}
which holds for any $A\in  L^2_{\rm\, loc}(\R^2)$ such that $\text{rot}\, A \not\equiv 0$, see \cite{lw,timo,kov2},
implies
\begin{equation*}
\beta=  \LL \psi_0, (i\nabla+A)^2\, \psi_0 \RR  = \| (i\nabla +A) \psi_0\|_2^2 >0.
\end{equation*}
We can thus apply Theorem \ref{thm-1} with $V=B$, and the claim follows.
\end{proof}

\begin{cor} \label{cor-amm3}
Let $\alpha =2$.  Let $B$ satisfy Assumptions \eqref{ass-B} and let $\beta$ be given by \eqref{beta-0-B}. Then
\begin{equation} \label{exp-dec-amm2}
\lim_{g \searrow 2}\, \big | \LL \psi_0, \, e^{-it P_\m(A,g)}\, \psi_0 \RR - e^{-it ( x_\eps -i \gamma_\eps)} \big|  =0
\end{equation}
uniformly in $t>0$, where $\eps= \frac{2-g}{2}, \ x_\eps =  \beta  \eps   \big[1+    \mathcal{O}\big(|\log \eps|^{-1}  \big)\big]$,  and
\begin{align*}
\gamma_\eps &= \frac{\eps\, | \LL   \psi_0, B\, \varphi_2 \RR|^2}{(\log \eps)^2}   \,    \big[1+    \mathcal{O}\big(|\log \eps|^{-1}
\big)\big]    \qquad\ \   \text{if} \quad  \LL   \psi_0, B\, \varphi_2 \RR \neq 0 , \\[6pt]
\gamma_\eps &= \frac{\eps^2}{4}   \Big(\,  \big | \LL  \psi_0,  B\, \varphi_1\RR  \big |^2 +  \frac{\pi^2 \beta^2}{\|e^h\|_2^2} \, \Big)
\qquad \qquad \quad \text{if} \quad  \LL   \psi_0, B\, \varphi_2 \RR = 0.
\end{align*}
\end{cor}

\begin{proof}
Since $\beta >0$, this is a combination of Theorems \ref{thm-2} and \ref{thm-3} applied with $V=B$.
\end{proof}

\section{\bf Resonances arising from a degenerate eigenvalue}
\label{sec-degenerate}

In this section we treat the resonances arising when the unperturbed operator has a degenerate eigenvalue at the threshold.
By \eqref{multiplicity},  zero is no longer a simple eigenvalue of $P(A)$ if $\alpha>2$. We need to isolate a zero-eigenfunction in order to
study
resonance behavior. It is natural to consider eigenfunctions of $P_0 VP_0$. Thus, we consider resonances arising as
\begin{equation} \nonumber
\LL \psi_0^{(V)}, \, e^{-it P_\m(A,\eps )}\, \psi_0^{(V)} \RR
\end{equation}
where $\psi_0^{(V)} \in P_0 L^2(\R^2)$ is a \emph{simple} eigenfunction of $P_0VP_0$. Similarly as in \cite{jn2}, in order to treat this case,
we
apply the SLFG formula twice: we first obtain the restriction of the resolvent to $P_0 L^2(\R^2)$ and then apply the formula again to obtain
the desired matrix element.

Explicitly, the operator $P_0VP_0$, as an operator on $P_0 L^2(\R^2)$, is an operator of rank at most $N$ (recall \eqref{multiplicity}), and
we work under the following assumption.
\begin{assumption} \label{ass-simple}
The operator $P_0 V P_0$ has a simple positive eigenvalue, $\kappa>0$, with (normalized) eigenfunction $\psi_0^{(V)}$. We let $P_1$ denote the
projection onto the subspace spanned by $\psi_0^{(V)}$ and $Q_1=P_0-P_1$.
\end{assumption}

\subsection{The application of the SLFG formula in the degenerate case}

As in the non-degenerate case, we take $\mathscr H_1= P_0 L^2(\R^2)$ and $\mathscr H_2=  Q_0 L^2(\R^2)$ so that (by \eqref{eq-resolvent})
\begin{equation} \nonumber
S(z)=\eps P_0VP_0-zP_0-\eps^2P_0vU\big(G(z)-\eps G(z) \left(U+\eps G(z)\right)^{-1}G(z)\big)UvP_0
\end{equation}
with $G(z)$ as defined in \eqref{G-eq}. For notational simplicity we let
\begin{equation} \label{S-tilde}
\widetilde{S}(z)=P_0vU\big(G(z)-\eps G(z) \left(U+\eps G(z)\right)^{-1}G(z)\big)UvP_0.
\end{equation}

We now write
\begin{equation} \nonumber
S(z)= \begin{pmatrix}
s_{11}(z) & s_{12}(z)  \\
s_{21}(z) & s_{22}(z)
\end{pmatrix} \\[3pt]
\end{equation}
where
\begin{equation} \nonumber
s_{11}(z)=\LL \psi_0^{(V)}, S(z) \psi_0^{(V)}\RR,
\end{equation}
so that, by the SLFG formula applied twice,
\begin{equation} \label{eq:SLFG-degenerate}
\begin{split}
&\LL \psi_0^{(V)}, \left(P_{\m}(A,\eps)-z \right)^{-1} \psi_0^{(V)}\RR =\LL \psi_0^{(V)}, S(z)^{-1} \psi_0^{(V)}\RR \\
&\quad =\left(s_{11}(z)-s_{12}(z) s_{22}(z)^{-1} s_{21}(z) \right)^{-1} \\
&\quad =\Big(\kappa \eps-z-\eps^2 \LL \psi_0^{(V)}, \widetilde{S}(z) \psi_0^{(V)} \RR-\eps^4\LL \psi_0^{(V)}, \widetilde{S}(z) Q_1 \left(\eps
Q_1 VQ_1-zQ_1-\eps^2 Q_1 \widetilde{S}(z)Q_1 \right)^{-1}Q_1 \widetilde{S}(z) \psi_0^{(V)} \RR \Big)^{-1} \\
&\quad = \widetilde{F}(z,\eps)^{-1}.
\end{split}
\end{equation}

\subsection{Non-integer flux}
For $\alpha \notin \Z$ we recall the decomposition \eqref{alpha-prime} and, with Proposition \ref{prop:degenerate-ni-exp} in mind, we denote
\begin{equation} \label{eq:Htilde}
\widetilde{H}(x)=\frac{\omega(1+\alp)  \, x^{\alpha'-1}}{1+ c_1 \omega(1+\alp) \, x^{\alpha'}}\  \psi \LL \psi,\cdot
\RR-\frac{\zeta(\alpha')\,
x^{-\alpha'}}{1+\varrho\, \zeta(\alpha') \, x^{1-\alpha'}}\ \varphi \LL \varphi,\cdot \RR.
\end{equation}
and
\begin{equation} \label{eq:Jtilde}
\widetilde{J}_\eps: = \Big( \frac{\kappa\eps}{2}\, ,\,   \frac{3\kappa\eps}{2}\Big )
\end{equation}

We also define the coefficients
\begin{equation} \label{eq:w34}
\w_3=\LL V \psi_0^{(V)}, \, \psi \RR, \qquad
\w_4=\LL V \psi_0^{(V)}, \varphi \RR,
\end{equation}
and
\begin{equation} \label{eq:tildenu}
\widetilde{\nu}=\left\{
\begin{array}{cc}
\mu  \, &\quad \text{if} \  \w_3 \neq 0\neq \w_4   , \\[4pt]
\alp  &\quad \text{if}  \ \ \w_3 \neq 0=\w_4 , \\[4pt]
1-\alpha' &\quad\ \  \text{if} \ \  \w_3=0 \neq \w_4 , \ ,
\end{array}
\right.
\end{equation}
where $\mu$ be defined by \eqref{mu}.
Finally, we let
\begin{equation} \label{eq:Ktilde}
\widetilde{K}(x,\eps)=\kappa \eps -x-\eps^2 \LL \psi_0^{(V)}, V\widetilde{H}(x)V \psi_0^{(V)} \RR.
\end{equation}

\begin{lem} \label{lem:FKtilde}
For $\eps>0$ sufficiently small
\begin{equation} \nonumber
\sup_{x\in \widetilde{J}_\eps} \big | \widetilde{F}(x,\eps)-\widetilde{K}(x,\eps) \big | \  \lesssim \    \eps^{1+2\widetilde{\nu}} +\eps^2 .
\end{equation}
\end{lem}
Here $\widetilde{F}(x,\eps)$ is defined in the same way as $F(x,\eps)$, see equation \eqref{F-of-x}.

\begin{proof}
Note that
\begin{equation} \nonumber
\begin{split}
\widetilde{F}(x,\eps)-\widetilde{K}(x,\eps)&=\eps^3 \LL \psi_0^{(V)}, vUG(x) \left(U+\eps G(x)\right)^{-1}G(x)Uv \psi_0^{(V)} \RR \\
&\quad -\eps^4\LL \psi_0^{(V)}, \widetilde{S}(x) Q_1 \big(\eps Q_1 VQ_1-xQ_1-\eps^2 Q_1 \widetilde{S}(x)Q_1 \big)^{-1}Q_1 \widetilde{S}(x)
\psi_0^{(V)} \RR +\mathcal{O}(\eps^2).
\end{split}
\end{equation}

We treat the two terms above separately. The first term is treated in the same way as in Lemma \ref{lem-FK}: for $\eps$ sufficiently small,
$(U+\eps
G(x))$ is boundedly invertible and so, by \eqref{eq:ResolventExpansionDeg},
\begin{equation} \nonumber
\begin{split}
\eps^3 \big| \LL \psi_0^{(V)}, vUG(x) \left(U+\eps G(x)\right)^{-1}G(x)Uv \psi_0^{(V)} \RR \big| & \lesssim \eps^3 \big
\|v\widetilde{H}(x)V \psi_0^{(V)} \big \|^2+\mathcal{O}(\eps^3)
\lesssim \eps^3 \big|\w_3 \eps^{\alpha'-1}+\w_4 \eps^{-\alpha'} \big|^2 \\[4pt] & \lesssim \eps^{1+2\widetilde{\nu}}+\eps^2.
\end{split}
\end{equation}
As for the second term, note that
$$
\eps^2 \sup_{x\in \widetilde{J}_\eps}  \|  \widetilde{S}(x) \|_{2\to 2} \lesssim\   \eps^{1+\mu}
$$
and that for sufficiently small $\eps$ and any $x \in \widetilde{J}_\eps$,
\begin{equation} \nonumber
\dist(x, \spec (Q_1VQ_1))=\dist\big (x, \spec(P_0 V P_0) \setminus \{\kappa\}\big) >\eps.
\end{equation}
Thus
\begin{equation} \label{Q1-eps}
\big \|  \big(\eps Q_1 VQ_1-xQ_1-\eps^2 Q_1 \widetilde{S}(x)Q_1 \big)^{-1} \big \|_{2\to 2}\  \lesssim\  \frac{1}{\eps}\, ,
\end{equation}
which implies
\begin{equation} \nonumber
\begin{split}
\eps^4 \, \big |\LL \psi_0^{(V)}, \widetilde{S}(x) Q_1 \big(\eps Q_1 VQ_1-xQ_1-\eps^2 Q_1 \widetilde{S}(x)Q_1 \big)^{-1}Q_1 \widetilde{S}(x)
\psi_0^{(V)} \RR \big | &\, \lesssim\,  \eps^3 \big \| \widetilde{S}(x) \psi_0^{(V)} \big \|^2. \\
\end{split}
\end{equation}
Using again the fact that $U+\eps G(x)$ is boundedly invertible for $\eps$ small enough, a straightforward computation shows that for such
$\eps$
\begin{equation} \nonumber
\begin{split}
\eps^3 \big\| \widetilde{S}(x) \psi_0^{(V)} \big \|^2 & \, \lesssim \, \eps^3  \big  \| V\widetilde{H}(x) V \psi_0^{(V)} \big \|^2+\eps^4
\big \|V\widetilde{H}(x) V\psi_0^{(V)} \big\| \big \|V\widetilde{H}(x)v(U+\eps G(z))^{-1}v \widetilde{H}(x) V \psi_0^{(V)} \big \| \\[4pt]
& \ \quad+\eps^5 \big \| V\widetilde{H}(x) v (U+\eps G(x))^{-1}v \widetilde{H}(x)\psi_0^{(V)}\big \|^2 +\mathcal{O}(\eps^3) \\[4pt]
& \lesssim\,  \eps^{1+2\widetilde{\nu}}+\eps^{1+3\widetilde{\nu}}+\eps^{1+4\widetilde{\nu}}+\eps^2 \lesssim \eps^{1+2\widetilde{\nu}}+\eps^2.
\end{split}
\end{equation}
This finishes the proof.
\end{proof}

If we now define
\begin{equation} \label{eta-sigma-tilde}
\widetilde \eta\,   =  \frac{ \pi^2\, |d|^2 }{4^{\alp}\, \Gamma^2(1+\alp)} \ , \qquad \ \widetilde \sigma\,   =
\frac{4^{\alp-1}}{\Gamma^2(1-\alp)} \, .
\end{equation}
and
\begin{equation} \label{g-tilde}
\widetilde g(x) = \frac{\widetilde \eta\,   |\w_3|^2\,  x^{\alp-1}}{|1+c_1 \omega(1+\alp) \, x^{\alp}|^2}\, +\,   \frac{\widetilde \sigma \,
|\w_4|^2\,  x^{-\alp}}{|1+\varrho\,  \zeta(\alp) \, x^{1-\alp}|^2} \, , \\[4pt]
\end{equation}
then we can state

\begin{thm} \label{thm-deg-1}
Let $2 <\alpha\not\in\Z$ and let $B$ satisfy Assumption \ref{ass-B}. Suppose moreover that $V \lesssim \langle \, \cdot\, \rangle ^{-\rho}$
for some $\rho>6$. Let $\psi_0^{(V)}$ be an eigenfunction of $\po V \po $ which satisfies Assumption \ref{ass-simple}, and assume that
\begin{equation} \label{w-34-con}
|\w_3|^2 +|\w_4|^2  >0.
\end{equation}
Then for sufficiently small $\eps$ we have
\begin{equation*}
\sup_{t>0} \big | \LL \psi_0^{(V)}, \, e^{-it P_\m(A,\eps )}\, \psi_0^{(V)} \RR - e^{-it ( x_\eps -i \gamma_\eps)} \big|  \, \lesssim\,
\eps^\mu,
\end{equation*}
where
$$x_\eps = \kappa\, \eps \big(1+ \mathcal{O}(\eps^{\widetilde\nu})\big)$$
and
\begin{equation} \label{deg-gamma}
\gamma_\eps = \eps^2\, \widetilde g(x_\eps).
\end{equation}
Recall that the
coefficients $\w_3$ and $\w_4$ are defined in \eqref{eq:w34}.
\end{thm}

\begin{proof}
Let
\begin{equation} \label{f-tilde}
\widetilde f(x) = \frac{ {\rm Re} (\omega(1+\alp)  ) \, x^{\alp-1} +|\omega(1+\alp)  |^2\,  x^{2\alp-1}}{|1+c_1 \omega(1+\alp)  \,
x^{\alp}|^2}\,
|\w_3|^2\, -    \frac{ {\rm Re} (\zeta(\alp) )\,   x^{-\alp} + \varrho\, |\zeta(\alp) |^2\, x^{1-2\alp} }{|1+\varrho\,  \zeta(\alp)\,
x^{1-\alp}|^2} \,
|\w_4|^2\, .
\end{equation}

In view of \eqref{gamma-function} and \eqref{zeta} we then have
\begin{equation} 
\begin{aligned}
\widetilde{K}_1(x,\eps) & = {\rm Re} \big(\widetilde{K}(x,\eps)\big) = \eps \kappa-x-\eps^2 \widetilde f(x) \\
\widetilde{K}_2(x,\eps) &= {\rm Im} \big(\widetilde{K}(x,\eps)\big) = -\eps^2\, \widetilde g(x)
\end{aligned}
\end{equation}
for all $x>0$. Note that in the non-degenerate case, when $1<\alpha <2$, we have $\alpha= 1+\alp$. Hence the functions $\widetilde g, \,
\widetilde f$
defined above have the same structure, with different coefficients, as their (non-degenerate) counterparts \eqref{g-eq}, \eqref{f-eq}. From
this point on we can thus repeat step by the step the arguments in the analysis of the non-degenerate case. In particular, a straightforward
modification of
Lemma \ref{lem-K-real} gives
\begin{equation}
\widetilde K_1(x,\eps) = \kappa \eps -x +  \mathcal{O}(\eps^{1+\widetilde\nu}).
\end{equation}
Furthermore, for $\eps$ small enough there exits a unique $x_\eps \in \widetilde J_\eps$ such that $\widetilde K_1(x_\eps ,\eps)=0$.
With this choice of $x_\eps$ we set
\begin{equation*}
\widetilde  I_\eps = \Big[ x_\eps - \frac{\kappa\eps}{4}, \, x_\eps + \frac{\kappa\eps}{4} \Big ] \, , \\[4pt]
\end{equation*}
and
\begin{equation*}
\widetilde p_\eps(t) =  \frac 1\pi \int_{\widetilde I_\eps} e^{-it x}\,   \, {\rm Im}\,\Big[\,  \frac{1}{\widetilde F(x, \eps)}\,  \Big]\,
dx\, .
\end{equation*}

Following the proofs of Lemmas \ref{lem-FtoK}-\ref{lem-p-eps} we then conclude that

\begin{equation*}
\sup_{t>0} \big | \widetilde p_\eps(t) -  e^{-it (x_\eps -i \gamma_\eps)}  \big |  \, \lesssim\, \eps^\mu.
\end{equation*}
holds for $\eps$ small enough. To prove the claim it remains to apply the Hunziker argument as in the proof of Theorem \ref{thm-1}, see in
particular
equations \eqref{eq-hunz-1}-\eqref{eq-hunz-2}.
\end{proof}

\smallskip

\begin{rem}
The constraint $\kappa>0$ in Assumption \ref{ass-simple} replaces condition \ref{cond-jn} in the non-degenerate case. Note also that
Theorem \ref{thm-deg-1} and equation \eqref{g-tilde} imply
\begin{equation*}
\gamma_\eps = \widetilde \eta\,  \kappa^{\alp-1} \, |\w_3|^2\,  \eps^{1+\alp} \big(1+ \mathcal{O}(\eps^\alp)\big) + \widetilde \sigma\,
\kappa^{-\alp} |\w_4|^2 \, \eps^{2-\alp}\, \big(1+ \mathcal{O}(\eps^{1-\alp})\big).
\end{equation*}
\end{rem}

\subsection{Integer flux}
 Given an eigenfunction $\psi_0^{(V)}$ of $\po V\po$ we denote
\begin{equation} \label{w6}
 \w_6=\LL V \psi_0^{(V)}, \, \varphi_2 \RR
\end{equation}

Similarly as in the case of the simple eigenvalue we first prove the result when $\w_6\neq 0$.

\begin{thm} \label{thm-deg-2}
Let $2 <\alpha\in\Z$ and let $B$ satisfy Assumption \ref{ass-B}. Suppose moreover that $V \lesssim \langle \, \cdot\, \rangle ^{-\rho}$
for some $\rho>6$. Let $\psi_0^{(V)}$ be an eigenfunction of $\po V \po $ which satisfies Assumption \ref{ass-simple}, and assume that
$\w_6\neq 0$.
Then for sufficiently small $\eps$ we have
\begin{equation*}
\sup_{t>0} \big | \LL \psi_0^{(V)}, \, e^{-it P_\m(A,\eps )}\, \psi_0^{(V)} \RR - e^{-it ( x_\eps -i \gamma_\eps)} \big|  \, \lesssim\,
|\log \eps|^{-1} ,
\end{equation*}
where $x_\eps = \kappa\, \eps \big(1+ \mathcal{O}(|\log \eps|^{-1} )\big)$ and where
$$
\gamma_\eps= \frac{\eps\, |\w_6|^2}{\kappa\, (\log \eps)^2}   \,    \big[1+    \mathcal{O}\big(|\log \eps|^{-1}  \big)\big] \, .
$$
\end{thm}

\begin{proof}
From equation \eqref{S-tilde} and the resolvent expansion \eqref{res-exp-2}  we deduce that for sufficiently small $\eps$,
$$
\sup_{x \in \widetilde{J}_\eps} \|  \widetilde{S}(x) \|_{2\to 2} \lesssim\   \frac{1}{\eps |\log\eps |} \,
$$
The resolvent equation in combination with \eqref{Q1-eps} thus implies that, as
bounded operators on $L^2(\R^2)$,
$$
 \big(\eps Q_1 VQ_1-xQ_1-\eps^2 Q_1 \widetilde{S}(x)Q_1 \big)^{-1} =  \big(\eps Q_1 VQ_1-xQ_1 \big)^{-1} + \mathcal{O}\Big(\frac{1}{\eps
 |\log\eps |}  \Big),
$$
where the error term is uniform in $x \in \widetilde{J}_\eps$. Hence writing $\Phi(x)  = Q_1 \widetilde{S}(x)  \psi_0^{(V)}$ it follows from
\eqref{eq:SLFG-degenerate} that
\begin{equation} \label{F-tilde-approx}
\widetilde{F}(x,\eps) = \kappa \eps-x-\eps^2 \LL \psi_0^{(V)}, \widetilde{S}(x) \psi_0^{(V)} \RR-\eps^4 k(x;\eps) +  \mathcal{O}\big(\eps\,
|\log\eps |^{-3} \big) ,
\end{equation}
where we have abbreviated
$$
k(x;\eps) = \LL \Phi(x), \left(\eps Q_1 VQ_1-xQ_1 \right)^{-1} \Phi(x) \RR \, .
$$
Notice that $k(x,\eps) \in \R$ for any $x>0$.
Similarly as in \eqref{fgh} we now define
\begin{equation} \label{fgh-hat}
\widehat\g(x) = \frac{|\w_6|^2}{x\, (\log x)^2} \, , \qquad\widehat f(x) = \frac{\widehat\g(x)}{\pi} \,  (\log x+m_\0), \qquad \widehat h(x) =
\frac{|\w_6|^2\, \LL  \varphi_2,
V  \varphi_2 \RR}{\pi^2 \, (x \log x)^2} ,
\end{equation} 
and $\widehat K(x;\eps) = \widehat K_1(x; \eps) + i \widehat K_1(x; \eps)$ with
\begin{align*}
\widehat K_1(x; \eps) &=  \kappa  \eps -x -\eps^2\, \widehat f(x) + \eps^3\,  \widehat h(x) -\eps^4 k(x;\eps) \\
\widehat K_2(x; \eps) &=  -\eps^2\, \widehat\g(x) .
\end{align*}
By equation \eqref{F-tilde-approx} we then have
\begin{equation} \label{F-tilde-K-hat}
\sup_{x\in \widetilde{J}_\eps} | \widetilde{F}(x,\eps)-\widehat{K}(x,\eps) | \  \lesssim \  \eps\, |\log\eps |^{-3} \, .
\end{equation}
This replaces the estimate of Lemma \ref{lem-F-exp}  from the non-degenerate case.
Obviously, equation \eqref{fgh-der} continues to hold with $f,\g, h$ replaced by the functions defined in \eqref{fgh-hat}. Moreover, an
elementary calculation shows that
\begin{equation*}
 \sup_{x\in \widetilde{J}_\eps} \eps^4\, | \partial_x^{j}\,  k(x;\eps) |     = \mathcal{O}(\eps^{1-j}\, |\log \eps|^{-2})\, , \qquad j=0,1,2.
\end{equation*}
We can thus mimic line by the line the rest of the proof of Theorem \ref{thm-2}, cf.~Lemmas \ref{lem-K-real-2}-\ref{lem-p-eps-int}, with
$x_\eps$ being the unique zero point of $\widehat K_1(x; \eps)$ in $ \widetilde{J}_\eps$ and with $\gamma_\eps= \eps^2\, \widehat\g(x_3)$.
This gives the claim.
\end{proof}

In order to find the leading term of $\gamma_\eps$ when the coefficient $\w_6$ vanishes we denote.
\begin{equation*}
\w_5=\LL V \psi_0^{(V)}, \varphi_1 \RR
\end{equation*}

We then have

\begin{thm} \label{thm-deg-3}
Let $2 <\alpha\in\Z$ and let $B$ satisfy Assumption \ref{ass-B}. Suppose moreover that $V \lesssim \langle \, \cdot\, \rangle ^{-\rho}$
for some $\rho>6$. Let $\psi_0^{(V)}$ be an eigenfunction of $\po V \po $ which satisfies Assumption \ref{ass-simple}. Assume that $\w_6 =0$,
and that
\begin{equation} \label{w-35-con}
|\w_3|^2 +|\w_5|^2  >0.
\end{equation}
Then
\begin{equation*}
\lim_{\eps\searrow 0}\,  \sup_{t>0} \big | \LL \psi_0^{(V)}, \, e^{-it P_\m(A,\eps )}\, \psi_0^{(V)} \RR - e^{-it ( x_\eps -i \gamma_\eps)}
\big|  =0,\\[4pt]
\end{equation*}
where $x_\eps = \kappa\, \eps \big(1+ \mathcal{O}(|\log \eps|^{-1}\big),$ and where
\begin{align*}
\gamma_\eps &= \frac{\eps^2}{4}   \Big(\,  |\w_5|^2 +  \frac{\pi^2 |\w_3|^2}{\|e^h\|_2^2} \, \Big)   .
\end{align*}
Recall that the
coefficients $\w_3$ is defined in \eqref{eq:w34}.
\end{thm}

\begin{proof}
We follow again the line of arguments used in the non-degenerate case.  Let
\begin{equation*}
\widetilde a_1 = \LL \psi_0^{(V)}, V \, \K \, V\, \psi_0^{(V)}\RR\, .
\qquad \widetilde a_2=  \LL \psi_0^{(V)}, V \, {\rm T}_0\, V\, \psi_0^{(V)}\RR\,
\end{equation*}
with $\K$ and $T_0$ given in Proposition \ref{prop:degenerate-int-exp}.
The self-adjointness of  $\K$ and ${\rm T}_0$ thus implies that $\widetilde a_1, \widetilde a_2 \in\R$.
Since $\w_6=0$, equations \eqref{eq:SLFG-degenerate} and \eqref{S-tilde} in combination with Proposition \ref{prop:degenerate-int-exp}
give
$$
\widetilde F(x,\eps) = \kappa \eps -x + \eps^2\, \widetilde a_1\, (\log x -i\pi) - \eps^2\, \widetilde a_2 + o(\eps^2),
$$
and
$$
\widetilde a_1 = \frac 14   \Big(\,  |\w_5|^2 +  \frac{\pi^2 |\w_3|^2}{\|e^h\|_2^2} \, \Big)\, .
$$
Hence if we define
\begin{equation}
\widehat K(x, \eps)  = \kappa \eps -x + \eps^2\, \widetilde a_1\, (\log x -i\pi) - \eps^2\, \widetilde a_2,
\end{equation}
then
$$
\sup_{x\in\widetilde J_\eps} \, | \widetilde F(x,\eps) - \widehat K(x, \eps)| = o(\eps^2).
$$
Setting $x_\eps$ to be the unique zero point of ${\rm Re} \big(\widehat K(x, \eps)\big) =  \kappa \eps -x + \eps^2\, \widetilde a_1\, \log x -
\eps^2\, \widetilde a_2$ in the interval
$\widetilde J_\eps$ and
$$
 \gamma_\eps = \widetilde a_1 \pi \eps^2
$$
we can thus follow the proof of Theorem \ref{thm-3} and conclude the proof of the theorem.
\end{proof}

\subsection{Example: $2<\alpha \leq 3$}
\label{ssec-example} Below we present an example in which a doubly degenerate eigenvalue of $P(A)$ produces two resonances which obey
different asymptotical
laws as $\eps\searrow 0$.

Let $B(x)$ satisfy Assumption \ref{ass-B} and assume that $B(x) = b(|x|)$ for some function $b:
\R_+\to \R$. Assume that
\begin{equation} \nonumber
2 \ < \ \frac{1}{2\pi} \int_{\R^2} B(x)\, dx\, \leq \ 3.
\end{equation}
Then $0$ is a doubly degenerate eigenvalue
of $P_\m(A)$ with normalized eigenfunctions
\begin{equation} \label{EF-double}
\psi_1(x) = c_1\, e^{h(x)}, \qquad \psi_2(x)= c_2\, (x_1+ix_2)\, e^{h(x)}\, .
\end{equation}
Since $h$ is radial in this case, we have $\LL \psi_1, \psi_2\RR=0$.
Moreover, by \cite[Cor.~5.10]{ko},
\begin{equation} \label{psi-not-L1}
\psi(x) =  \frac{(x_1+ix_2)\, e^{h(x)}}{\|(x_1+ix_2)\, e^{h}\|_2^2}\, .
\end{equation}
Now let
\begin{equation}
V(x) = (1+x_1^2-x_2^2)\, u(|x|)\, ,
\end{equation}
where $u: \R_+\to \R$ satisfies $u\geq 0$ and $u \lesssim \LL \cdot \RR^{-\rho}$ with some $\rho>8$. We then have
\begin{equation}
\po V\, \po = \begin{pmatrix}
\kappa_1\, \psi_1 \LL \psi_1\, ,\, \cdot\,  \RR& 0  \\
0 & \kappa_2\, \psi_2 \LL \psi_2\, ,\, \cdot\,  \RR
\end{pmatrix} \ ,
\end{equation}
where
$$
\kappa_1 = \LL \psi_1, V \psi_1\RR >0, \qquad \kappa_2 = \LL \psi_2, V \psi_2\RR >0 .
$$
The last two inequalities follow from the positivity of $u$ and from the fact that
$$
\int_{\R^2} (x_1^2-x_2^2)\, u(|x|)\, |\psi_j(x)|^2\, dx = 0 \qquad j=1,2.
$$
Thus $\psi_1$ and $\psi_2$ are eigenfunctions of $\po V\, \po$ with eigenvalues $\kappa_1>$ and $\kappa_2>0$. Moreover,
by choosing $u$ in a suitable way, for example putting $u=0$ in a ball centered in the origin with a radius large enough, we can make sure
that
$$
\kappa_1 = \int_{\R^2} u(|x|)\, |\psi_1(x)|^2\, dx \, <  \, \int_{\R^2} u(|x|)\, |\psi_2(x)|^2\, dx = \kappa_2,
$$
cf.~\eqref{EF-double},  so that
Assumption \ref{ass-simple} is satisfied for both $\psi_1$ and $\psi_2$. Now we treat the following cases:

\smallskip

\underline{$2<\alpha <3$}.  By \cite[Cor.~5.10]{ko},
$$
\varphi(x) = (x_1+ix_2)^2\, e^{h(x)}\, .
$$
A short calculation shows that
\begin{align}
\w_3^{(1)}  & = \LL V \psi_1\, ,\,  \psi \RR =0, \qquad \w_3^{(2)}  = \LL V \psi_2\, ,\,  \psi \RR  \neq 0  \label{w3-example} \\
\w_4^{(2)} & =\LL V \psi_2\, , \, \varphi \RR=0, \qquad \w_4^{(1)} =\LL V \psi_1\, , \, \varphi \RR \neq 0. \label{w4-example}
\end{align}
Condition \eqref{w-34-con} is thus satisfied and
Theorem \ref{thm-deg-1} can be applied with $\psi_0^{(V)}=\psi_1$ as well as with $\psi_0^{(V)}=\psi_2$. Hence $P_\m(A,\eps)$ admits two
resonances which satisfy
\begin{equation} \label{doubly-not-int}
\begin{aligned}
x_\eps^{(1)} & = \eps \kappa_1 +\mathcal{O}(\eps^{2-\alp}) \, \qquad \gamma_\eps^{(1)} = \widetilde\sigma \, \big|\w_4^{(1)}\big |^2\,
\kappa_1^{-\alp}\, \eps^{2-\alp} + \mathcal{O}(\eps^{1+2\alp})\\[4pt]
x_\eps^{(2)} & = \eps \kappa_2 +\mathcal{O}(\eps^{1+\alp}) \, \qquad \gamma_\eps^{(2)} = \widetilde\eta \, \big|\w_3^{(2)} \big|^2\,
\kappa_2^{\alp-1}\, \eps^{1+\alp} + \mathcal{O}(\eps^{3-2\alp})\
\end{aligned}
\end{equation}

\medskip

\underline{$\alpha =3$}. Obviously, the coefficients $\w_3^{(j)}, j=1,2,$  still satisfy \eqref{w3-example}.
The zero virtual bound states are given by
$$
\varphi_1(x) = (x_1+ix_2)^3\, e^{h(x)}\, , \qquad \varphi_2(x) = (x_1+ix_2)^2\, e^{h(x)},
$$
see \cite[Cor.~6.7]{ko}. A quick calculation thus gives
$$
\w_6^{(1)} = \LL V \psi_1\, , \, \varphi_2 \RR \neq 0\, .
$$
Hence Theorem \ref{thm-deg-2} applies and says that
\begin{equation*}
\sup_{t>0} \big | \LL \psi_1 \, e^{-it P_\m(A,\eps )}\, \psi_1 \RR - e^{-it ( x_\eps^{(1)} -i \gamma_\eps^{(1)})} \big|  \, \lesssim\,
|\log \eps|^{-1}
\end{equation*}
with $x_\eps^{(1)} = \kappa_1\, \eps \big(1+ \mathcal{O}(|\log \eps|^{-1} )\big)$ and with
\begin{equation} \label{gam-1}
\gamma_\eps^{(1)}= \frac{\eps\, \big |\w_6^{(1)}\big|^2}{\kappa_1 (\log \eps)^2}   \,    \big[1+    \mathcal{O}\big(|\log \eps|^{-1}
\big)\big] \, .
\end{equation}
On the other hand, for the eigenfunction $\psi_2$ we have
$$
\w_6^{(2)} = \LL V \psi_2\, , \, \varphi_2 \RR  = \w_5^{(2)} = \LL V \psi_2\, , \, \varphi_1 \RR =0.
$$
This means that Theorem \ref{thm-deg-2} is not applicable to $\psi_2$. However, in view of equation
\eqref{w3-example} we can apply Theorem \ref{thm-deg-3} which implies that
$$
\lim_{\eps\searrow 0} \sup_{t>0} \big | \LL \psi_2, \, e^{-it P_\m(A,\eps )}\, \psi_2 \RR - e^{-it ( x_\eps^{(2)} -i \gamma_\eps^{(2)})} \big|
= 0,
$$
with $x_\eps^{(2)} = \kappa_2\, \eps \big(1+ \mathcal{O}(|\log \eps|^{-1} )\big)$ and
\begin{align} \label{gam-2}
\gamma_\eps^{(2)} &= \eps^2\, \frac{\pi^2 \big|\w_3^{(2)}\big|^2}{\|e^h\|_2^2} \, .
\end{align}

\smallskip

\begin{rem} \label{rem-degenerate}
In general, given an arbitrary $\alpha>1$, we note that if all the eigenvalues of the operator $\po V \po$ are positive and simple, and if
the associated eigenfunctions satisfy condition \eqref{w-34-con} resp.~\eqref{w-35-con}, then the number of resonances of $P_\eps(A)$
is at least as large as the multiplicity of $0$ as an eigenvalue of $P(A)$.
\end{rem}

\appendix

\section{\bf }
\label{sec-app}

\begin{proof}[Proof of Propositions \ref{prop:degenerate-ni-exp} and \ref{prop:degenerate-int-exp}]
Although not explicitly stated there, the existence of the operator $S_0$ in \eqref{eq:ResolventExpansionDeg} follows from the proof of
\cite[Thm.~5.6]{ko}.
The self-adjointness of $S_0$ follows by considering the expansion \eqref{eq:ResolventExpansionDeg} for $z=-x$ with $x>0$. Then $(P_-(A)
-z)^{-1}= (P_-(A)
+x)^{-1}$ is self-adjoint. On the other hand, $\omega(1+\alp) (-x)^{\alp-1}$ and $\zeta(\alp) (-x)^{-\alp}$ are real numbers, see equation
\eqref{zeta}. Hence the first three operators on the right hand side of \eqref{eq:ResolventExpansionDeg}  are self-adjoint. Therefore $S_0$ is
self-adjoint
too. In the same way  one proves the existence and self-adjointness of ${\rm T_0}$ in
Proposition \ref{prop:degenerate-int-exp}.
\end{proof}


\section{\bf }
\label{sec-app-2}

The following technical result is not new, see \cite[Eq.~(3.56)]{jn}. For the sake of completeness, we give a short proof.

\begin{lem} \label{lem-jn}
Let $a>1$. Then
\begin{equation*}
\sup_{s>0} \, \sup_{R >a}\,   \Big | \int_{-R}^R \frac{e^{-i s \frac yR}}{y+i}\, dy\,  \Big | \, < \, \infty.
\end{equation*}
\end{lem}

\begin{proof}
By the residue theorem,
$$
\int_{-R}^R \frac{e^{-i s \frac yR}}{y+i}\, dy\ = -2\pi i \, e^{-\frac s R}  - \int_\Gamma  \frac{e^{-i s \frac zR}}{z+i}dz
$$
where $\Gamma$ is the semi-circle of radius $R$ in the lower complex half-plane centered in the origin and directed from $(R,0)$ to $(-R,0)$.
The first term on the hand side is obviously bounded
uniformly in $s>0$ and $R>a$. As for the second term, parametrizing $\Gamma$ in the usual way we find
\begin{align*}
 \Big |   \int_\Gamma  \frac{e^{-i s \frac zR}}{z+i}dz\,  \Big | & =  R\  \Big | \int_0^\pi  \frac{e^{-i s \cos\theta -s
 \sin\theta}e^{i\theta}}{R \cos \theta -i R\sin\theta  +i}\, d\theta\, \Big | \leq \frac{\pi R}{R-1}\leq \frac{\pi a}{a-1}\, ,
\end{align*}
since $|R \cos \theta -i R\sin\theta  +i | \geq R-1$.
\end{proof}

\end{document}